\documentclass{amsart}

\usepackage{hyperref,array}
\usepackage[utf8]{inputenc}
\usepackage[english]{babel}
\usepackage{amsmath,amssymb}
\usepackage{enumitem} %[label=\alph*)]
\usepackage{mathtools}
\usepackage{amsthm}
\usepackage{thmtools}

\usepackage{calc} %\widthof

\usepackage{tikz}
\usetikzlibrary{automata,arrows,patterns,arrows.meta}
\tikzstyle{dot}=[circle,draw,minimum size=0.8mm,inner sep=0pt,fill]

\declaretheorem[numberwithin=section, name=Lemma]{lemma}

\declaretheorem[name=Theorem,sibling=lemma]{theorem}
\declaretheorem[name=Corollary,sibling=lemma]{corollary}
\declaretheorem[name=Conjecture,sibling=lemma]{conjecture}
\declaretheorem[style=definition, name=Example,sibling=lemma]{example}
\declaretheorem[style=definition,name=Definition,sibling=lemma]{definition}

\newcommand{\N}{\mathbb{N}}

\newcommand{\define}[1]{\emph{#1}}
\newcommand{\notiff}{%
  \mathrel{{\ooalign{\hidewidth$\not\phantom{"}$\hidewidth\cr$\iff$}}}}

\newcommand{\bound}[1]{B_{\text{#1}}}%TODO rename this
\newcommand{\Dyck}{\textbf{Dyck}}
\newcommand{\Prime}{\textbf{Prime}}

\newcommand{\UInfty}{\textbf{UInf}}
\newcommand{\Pali}{\textbf{Pali}}
\newcommand{\DetPali}{\textbf{DPali}}
\newcommand{\T}{\textbf{T}}
\newcommand{\B}{\textbf{T'}}
\newcommand{\PAD}{\operatorname{PAD}}
\newcommand{\ind}{\operatorname{ind}}

\newcommand{\enc}{\operatorname{enc}}

\newcommand{\im}{\operatorname{im}}
\newcommand{\core}{\operatorname{core}}
\newcommand{\head}{\operatorname{head}}
\newcommand{\bottom}{\operatorname{bottom}}
\newcommand{\front}{\operatorname{front}}
\newcommand{\tail}{\operatorname{tail}}

\newcommand{\Top}[1]{\mathcal{T}_\mathcal{#1}}
\newcommand{\Pos}{P}

\setlist{noitemsep}

\title{Exploring the Topological Entropy of Formal Languages}
\author{Florian Starke}


\begin{document}

\begin{abstract}
We introduce the notions of \define{topological entropy} of a formal language and of a topological automaton. We show that the entropy function is surjective and bound the entropy of languages accepted by deterministic $\varepsilon$-free push-down automata with an arbitrary amount of stacks.
\end{abstract}

\maketitle

\section{Introduction}

A well established notion to measure the complexity of a dynamical system is \define{topological entropy}. It measures how chaotic or random a dynamical system is.
A topological automaton contains a dynamical system.
Using this we can define the complexity of a language to be the topological entropy of the minimal topological automaton accepting it.

%previous work
Steinberg introduced the notion of a topolocial automaton in 2013 \cite{steinberg}. Then in 2016, Schneider and Borchmann used this notion to define the topological entropy of a formal language. They gave a characterization of the topological entropy of a formal language in terms of Myhill-Nerode congruence classes and determined the entropy of some example languages \cite{schneiderborchmann}.

%new results
In this article we solve previously open problems in this field, and expand the variety of example languages. In particular, we show that the entropy function is surjective. 
We show that every language accepted by a deterministic $\varepsilon$-free counter automaton with an arbitrary amount of counters has zero entropy, which generalizes the fact that all regular languages have zero entropy. %as conjectured in [daniel]
Furthermore, we give a finite upper bound for the entropy of languages accepted by deterministic $\varepsilon$-free push-down automata with an arbitrary amount of stacks, which shows that every deterministic $\varepsilon$-free context-free language has finite entropy. 
We determine the entropy of the Dyck languages, the deterministic palindrome language, and of some other new example languages. Among them is also a deterministic context-free language with infinite entropy.

This article is structured as follows. First we introduce the notion of a topological automaton and its topological entropy. 
We will have a brief discussion on the effect different encodings have on the entropy of a language in Section \ref{sec:Encoding}.
This will motivate Section \ref{sec:Surjectivity}, where we will show that for every possible entropy there is also a language with that entropy.
In Section~\ref{sec:DecisionProblems} we will take a glimpse on the connection between topological entropy and complexity theory by calculating the entropy of \textbf{SAT} and looking at the effect padding has on the entropy.
In Section~\ref{cha:automata} we will bound the entropy of languages accepted by certain kinds of counter automata and push-down automata. There we will also explore the connection between topological entropy and the Chomsky hierarchy.
In the end we will give an outline for future work.

\section{Preliminaries}\label{sec:StateOfTheArt}

In this section we will give a brief introduction to the Myhill-Nerode congruence relation, topological automata, and their topological entropy. 
We will use the results from Schneider and Borchmann \cite{schneiderborchmann} to give a compact definition of the topological entropy of a formal language.

\subsection{Myhill-Nerode Congruence Relation}
The Myhill-Nerode congruence relation is a very basic concept of formal languages, but as it is essential to this article we will give a short recapitulation.
Let $L$ be a formal language over some alphabet $\Sigma$. The \define{Myhill-Nerode right-congruence relation of $L$}, denoted by $\Theta(L)$, is %\FlorianSagt{ist das ok so oder soll ich rechts und links kongruenz auch nochmal erwähnen?}
\[\Theta(L)=\{(u,v)\mid \forall w\in\Sigma^*.\ uw\in L\iff vw\in L\}.\]
For two words $u$ and $v$ we say that $w$ \define{witnesses} $(u,v)\notin\Theta(L)$ if $uw\in L\notiff vw\in L$. For a word $u$, we call all words $w$ with $uw\in L$ \define{positive witnesses of $u$}. Denote the set of all positive witnesses of $u$ by $\Pos_u$. % and all words $w$ with $uw\notin L$ \define{negative witnesses of $u$}.
Note that $[u]=[v]$ if and only if $\Pos_u=\Pos_v$.

If a language $L$ has only finitely many congruence classes, then the minimal finite automaton accepting it can be constructed using $\Theta(L)$: 
take the set $\Sigma^*/\Theta(L)$ as states, where the final states are the congruence classes contained in $L$, $[\varepsilon]$ as initial state, and  $([w],a)\mapsto[wa]$ as transition function.

%_________________________________________________________________________________________________
\subsection{Topological Automata and Topological Entropy}

Every deterministic finite automaton $\mathcal{A}$ contains a dynamical system, where the monoid $\Sigma^*$ acts on the states of $\mathcal{A}$. In the following definition we generalize this idea to automata with infinitely many states.
\begin{definition}
A \define{topological automaton} is a 5 tuple $\mathcal{A}=(X,\Sigma,\delta,x_0,F)$ where
\begin{itemize}
\item $X$ is a compact Hausdorff space (the \define{states}),
\item $\Sigma$ is an alphabet (the \define{input alphabet}),
\item $\delta\colon X\times\Sigma^*\to X$ is a continuous action of $\Sigma^*$ on $X$ (the \define{transition function}),
\item $x_0$ is an element from $X$ (the \define{initial state}), and
\item $F$ is a clopen, i.e., closed and open, subset of $X$ (the \define{final states}).
\end{itemize}
The language accepted by $\mathcal{A}$ is
\[L(\mathcal{A})=\{w\in\Sigma^*\mid \delta(x_0,w)\in F\}.\]
\end{definition}

We can define the topological entropy of $\mathcal{A}$ as the entropy of the underlying dynamical system but we do not want to introduce all the necessary notions and therefore we will use an equivalent definition.

Let $L=L(\mathcal{A})$, and $E\subseteq \Sigma^*$. We define the following two equivalence relations
\begin{align*}
\Theta_E(L)&{}=\{(u,v)\mid \forall w\in E.\ uw\in L \iff vw\in L\}\text{ and}\\
\Lambda_E(\mathcal{A})&{}=\{(x,y)\mid \forall w\in E.\ \delta(x,w)\in F\iff \delta(y,w)\in F\}.
\end{align*}
The relation $\Theta_E(L)$ is an approximation of the Myhill Nerode congruence relation of $L$ in the sense that it allows only words from $E$ as witnesses. Hence if we choose $E=\Sigma^*$, then $\Theta_{\Sigma^*}(L)$ is the Myhill Nerode congruence relation of $L$.
The second relation $\Lambda_E(\mathcal{A})$ is the counterpart of $\Theta_E(L)$ for the states of $\mathcal{A}$. 

For $n\in\N$ we denote 
\[\Theta_{\Sigma^{(n)}}(L)\text{ by }\Theta_n(L)\text{ and }\Lambda_{\Sigma^{(n)}}(\mathcal{A})\text{ by }\Lambda_n(\mathcal{A}).\] 
We will write $\Pos_u^{(n)}$ (and $\Pos_u^{n}$) for the positive witnesses of $u$ with length at most $n$ (with length exactly $n$).
Schneider and Borchmann showed (Lemma 3.6 from \cite{schneiderborchmann}) that if $\mathcal A$ is a topological automaton, where the reachable states are dense in the set of all states. Then the \define{topological entropy} $\eta(\mathcal{A})$ of $\mathcal A$ satisfies:

\[\eta(\mathcal{A})=\limsup_{n\to\infty}\frac{\log_2(\ind \Lambda_n(\mathcal{A}))}{n}.\]

Note that since $\Sigma^{(n)}$ is finite the index of $\Lambda_n(L)$ is also finite.
Since from every topological automaton we can obtain a topological automaton accepting the same language, where the reachable states are dense in the set of all states, we will use the above identity as a definition.
Analogously, we define the \define{topological entropy} of a language:

\[\eta(L)=\limsup_{n\to\infty}\frac{\log_2(\ind \Theta_n(L))}{n}.\]

Furthermore, Schneider and Borchmann give us the following connection between the entropy of an automaton and the entropy of its language.
\begin{theorem}[Theorem 3.10 from \cite{schneiderborchmann}]
Let $L$  be a language. Then
\[\eta(L)=\min\{\eta(\mathcal{A})\mid \mathcal{A}\text{ is a topological automaton with }L(\mathcal{A})=L\}.\]
\end{theorem}
In particular, for every language $L$ there also exists a topological automaton $\mathcal{A}$ such that $\eta(L)=\eta(\mathcal{A})$.

%_________________________________________________________________________________________________
\section{Encodings and Surjectivity}
In this section we will first discuss how encodings effect the entropy of a language. In Section \ref{sec:Surjectivity} we will solve the previously open problem of whether the entropy function is surjective by constructing a language for every possible entropy.
Finally, we will discuss the entropy of languages over unary alphabets.

\subsection{Encoding}\label{sec:Encoding}

We motivate this subsection with an example.
\begin{example}\label{exa:Dyck}
An example considered in \cite{schneiderborchmann} is the \define{Dyck language with $k$ sorts of parenthesis}, which consists of all balanced strings over $\{(_1,)_1,\dots,(_k,)_k\}$.
More generally, let $\Gamma$ be an alphabet and $\overline{\Gamma}=\{\overline{a}\mid a\in\Gamma\}$. Then $\overline{\phantom{a}}\colon\Gamma\to\overline{\Gamma}$ is a bijection. %, we denote the inverse of $\overline{\phantom{a}}$ again by $\overline{\phantom{a}}$. 
Now the \define{Dyck language over $\Gamma$}, denoted by $\Dyck_\Gamma$, is the set of all words $u$ such that successively replacing $a\overline{a}$ in $u$ by $\varepsilon$ results in $\varepsilon$.
\end{example}

Later in Lemma \ref{lem:Dyck} we will show that $\eta(\Dyck_\Gamma)=\log_2|\Gamma|$. This shows that there are Dyck languages with arbitrarily high entropy. But what happens with the entropy if we encode all Dyck languages over a two element alphabet? 
In this section we will give upper and lower bounds for the entropy of an encoded language. Then we will apply these results to the Dyck languages, to show that their entropy is bounded if we encode them over a fixed alphabet.

\begin{definition}
An encoding of $\Sigma$ over $\Gamma$ is a mapping $\enc\colon\Sigma\to\Gamma^+$ with the prefix property, i.e., there is no word in the image of $\enc$ which is a prefix of another word in the image. 
\end{definition}

The prefix property is necessary to ensure the invertability of the encoding.%because otherwise the encoding might not be invertible.
%For now we will restrict ourselves to encodings where all encoded words have the same length and we will generalize the result afterwards.

\begin{lemma}
Let $L$ be a language over $\Sigma$ and $\enc\colon\Sigma\to\Gamma^+$ an encoding of $\Sigma$ over $\Gamma$. Then $\frac{\eta(L)}{k_1}\leq \eta(\enc(L))\leq \frac{\eta(L)}{k_2}$, where $k_1=\max\{|u|\mid u\in\im(\enc)\}$ and $k_2=\min\{|u|\mid u\in\im(\enc)\}$.
\end{lemma}
\begin{proof}
Note that $\varepsilon$ is not in the image of $\enc$, and as a consequence $k_2$ is at least 1.
We show the following two inequalities:
\begin{align*}
\ind\Theta_{n\cdot k_1}(\enc(L))&\geq \ind\Theta_{n}(L) \tag{$*$}\\
\ind\Theta_{n\cdot k_2}(\enc(L))&\leq |\textit{Pre}|\cdot\ind\Theta_{n}(L)+1 \tag{$**$}
\end{align*}
where $\textit{Pre}$ contains all real prefixes of words in the image of $\enc$.

For the first inequality consider the map
\[[u]\mapsto\{[\enc(u')]\mid u'\in[u]\}.\]
Note that injectivity of this map does not suffice to show $(*)$. We need to show that the images of two different classes $[u_1]$ and $[u_2]$ are disjoint. 
%But to show $(*)$ it suffices to show that all sets in the image of $[w]\mapsto\{[\enc(w')]\mid w'\in[w]\}$ are pairwise disjoint. 
Let $[\enc(u_1')]$ and $[\enc(u_2')]$ be elements from sets corresponding to $[u_1]$ and $[u_2]$, respectively.
Because $u_1'\in[u_1]$ and $u_2'\in[u_2]$ there is a $w$ that witnesses $(u'_1,u'_2)\notin\Theta_{n}(L)$ and $|\enc(w)|\leq k_1\cdot|w|\leq k_1\cdot n$. Hence $\enc (w)$ is a witness for $(\enc (u'_1),\enc (u'_2))\notin\Theta_{n\cdot k_1}(\enc (L))$. %Note that this is the injectivity part from before, we have just dropped the well definedness.

For $(**)$ we will show that the following map is almost surjective
\begin{align*}
\Sigma^*/{\Theta_n(L)}\times \textit{Pre}&\to \Gamma^*/{\Theta_{n\cdot k_2}(\enc(L))}\\
([u],v)&\mapsto [\enc(u)\cdot v].
\end{align*}
%Note that every word in $\Gamma^*$ has at most one decomposition $uv$ with $u\in \enc(\Sigma^*)$ and $v\in A$. 
Note that all words $u\in\Gamma^*$ with $\Pos_u=\emptyset$ lie in the same class. Let $[u]$ be a class of $\Theta_{n\cdot k_2}(\enc(L))$ such that $u$ has at least one positive witness $w$. Then $uw$ is the encoding of some word in $\Sigma^*$. Therefore there are $u'\in\Sigma^*$ and $v\in \textit{Pre}$ such that $\enc(u')v=u$. As a consequence, $([u'],v)$ lies in the preimage of $[u]$. From this $(**)$ follows.

Using these inequalities we can now infer
\begin{align*}
k_1\cdot \eta(\enc(L))
=k_1\cdot\limsup_{n\to\infty}\frac{\log_2\operatorname{ind}\Theta_{n\cdot k_1}(\enc(L))}{n\cdot k_1}
\stackrel{(*)}\geq \limsup_{n\to\infty}\frac{\log_2(\ind\Theta_{n}(L))}{n}
= \eta(L).
\end{align*}

The other direction follows similarly 
\begin{align*}
k_2\cdot \eta(\enc(L))%=k_2\cdot\limsup_{n\to\infty}\frac{\log_2\operatorname{ind}\Theta_{n\cdot k_2}(\enc(L))}{n\cdot k_2}
\stackrel{(**)}\leq \limsup_{n\to\infty}\frac{\log_2( |\Gamma^{(k_2-1)}|\cdot\ind\Theta_{n}(L)+1)}{n}= \eta(L).
\end{align*}

This concludes the proof.
\end{proof}

Note that if a language has infinite entropy, then every encoding of this language also has infinite entropy. On the other hand if $\eta(L)$ is zero, then the entropy of $\enc(L)$ will also be zero.
If all encoded letters have the same length we obtain the following corollary.

\begin{corollary}\label{cor:encSameLength}
Let $L$ be a language over $\Sigma$, $k\geq 1$, and $\enc\colon\Sigma\to\Gamma^k$ an encoding of  $\Sigma$ over $\Gamma$. Then $\eta(\enc(L))=\frac{\eta(L)}{k}\,$.
\end{corollary}

We can encode every language over $\Sigma$ over a two element alphabet $\Gamma$ with an encoding $\enc\colon\Sigma\to\Gamma^k$ for some $k\in\N$. Note that the minimal possible value of $k$ for which we can define such an encoding is $\lceil\log_2|\Sigma|\rceil$. Hence we call the encoding $\enc$ \define{efficient} if $k=\lceil\log_2|\Sigma|\rceil$.
\begin{corollary}\label{cor:encBinary}
If $L$ is a language over $\Sigma$ and $\enc\colon\Sigma\to\Gamma^k$ efficiently encodes $\Sigma$ over a two letter alphabet $\Gamma$, then $\eta(\enc(L))=\frac{\eta(L)}{\lceil\log_2|\Sigma|\rceil}$.
\end{corollary}
%Note that the minimal possible $k$ is $\lceil\log_2|\Sigma|\rceil$.\DanielSagt{Das klingt komisch \dots{} was genau ist mit \enquote{minimal possible $k$} gemeint?}

Now we can answer our initial question: What happens if we encode the Dyck languages over a two element alphabet?
\begin{corollary}
Let $\enc\colon\Sigma\to\Delta^k$ be an efficient encoding of $\Sigma=\Gamma\cup\overline{\Gamma}$ over a two letter alphabet $\Delta$. Then
\[\eta(\enc(\Dyck_\Gamma))=\frac{\log_2|\Gamma|}{\lceil\log_2|\Gamma|\rceil+1}.\]
\end{corollary}
For example, if we encode $\Dyck_{\{(_1,(_2\}}$ over $\{0,1\}$ with
\begin{align*}
(_1&\mapsto 00 & )_1&\mapsto 10\\
(_2&\mapsto 01 & )_2&\mapsto 11
\end{align*}
then the encoded language has entropy $\frac{1}{2}$.

This corollary implies that any encoded Dyck language has entropy less than one.
As this argument also affects the other example presented in \cite{schneiderborchmann}, namely the palindrome languages, we now lack examples for languages over a two element alphabet with an entropy in $(2,\infty)$. 
We will remedy this in the next section by constructing a language for any given entropy.

\subsection{Every Entropy has its Language}\label{sec:Surjectivity} %Every Jack has his Jill

In this subsection we will show that the entropy function $\eta\colon\mathcal{P}(\Sigma^*)\to [0,\infty]$ is surjective if $\Sigma$ contains at least two elements.
For now we fix the alphabet $\Sigma=\{0,1\}$. 

Let us call a sequence $(k_n)_{n\in\N}$ of natural numbers \define{suitable} if it is monotone increasing, $k_0=1$, and $k_n\leq 2\cdot k_{n-1}$. Note that in this case $k_n\leq 2^n$. Our goal is to construct a language $L$ such that $\ind\Theta_n(L)$ is about $2^{k_n}$. The construction generalizes an idea from Schneider and Borchmann (Example 4.12 in \cite{schneiderborchmann}).
We define for all $n\in\N$
\begin{align*}
\varphi_n\colon \Sigma^{2^n}&\to \mathcal{P}(\{1,\dots,k_n\})\\
a_1\dots a_{2^n}&\mapsto\{i\in\{1,\dots,k_n\}\mid a_i=1\}.
\end{align*}
Note that $|\im\varphi_n|=2^{k_n}$.
For $n\in\N$ we define a function $f_n\colon\Sigma^n\to\{1,\dots,k_n\}$ recursively.
Let us fix $f_0(\varepsilon)=1$. For $n\in\N$ define
\begin{align*}
f_{n+1}(0u)=f_{n}(u) && f_{n+1}(1u)=
	\begin{cases}
	f_{n}(u) & \text{if } f_{n}(u) + k_{n}>k_{n+1}\\
	f_{n}(u) + k_{n} & \text{if } f_{n}(u) + k_{n}\leq k_{n+1}.
	\end{cases}
\end{align*}

We construct a language $L$ in the following way:
\[L=\left\{uv~\middle|~ |v|=2^{|u|},  f_{|u|}(v)\in\varphi_{|u|}(u)\right\}.\]

%Before we continue, let us make some observations. Firstly, note that $|\im\varphi_n|=2^{k_n}$. 

The idea of $f_n$ is that if $k_{n+1}<2\cdot k_n$, then we do not need all possible new words to distinguish all elements in $\im\varphi_{n+1}$. This is because the number of words doubles since every word $u$ is split into $0u$ and $1u$. As a consequence, depending on the way we look at it, $f_n$ fuses some of these words together, or only splits as much words as needed. %TODO elaborate on connection between f_n and f_n+1 (if f_n fuses to words together then so does f_n+1)

Now it is time to take a closer look at the properties of $\varphi_n$ and $f_n$. Clearly, $\varphi_n$ and $f_n$ are surjective. More interesting are the following two properties:

\begin{enumerate}[label=(P\arabic*)]
\item\label{enu:P1} the following are equivalent for all $n\in\N$ and all $u,v\in\Sigma^n$
\begin{enumerate}
\item $f_n(u)=f_n(v)$
\item $f_{n+1}(au)=f_{n+1}(av)$ for all $a\in\Sigma$
\item $f_{n+1}(au)=f_{n+1}(av)$ for some $a\in\Sigma$%$\pi_{2,\dots,n+1}(\ker f_{n+1})=\ker f_n$ for
\end{enumerate}
%\item $\varphi_n(u)=\varphi_n(v)$ iff $\varphi_{n+1}(uw)=\varphi_{n+1}(vw)$, and %pr(phi(wu))=phi(w)
%\item NEW FIX (TOO STRONG ?) $f_{n+1}(au)=f_{n+1}(bv)$ implies $f_n(u)=f_n(v)$
\item\label{enu:P2} the following are equivalent for all $n,k\in\N$, $u_1,u_2\in\Sigma^k$ with $k\leq 2^n$:
\begin{enumerate}
\item $\varphi_n(u_1v)=\varphi_n(u_2v)$ for some $v\in\Sigma^{2^n-k}$
\item  $\varphi_n(u_1v)=\varphi_n(u_2v)$ for all $v\in\Sigma^{2^n-k}$.
\end{enumerate}
\end{enumerate}

Let us check that these properties hold. For \ref{enu:P1} consider first $\text{(a)}\Rightarrow\text{(b)}$. We either have 
\begin{align*}
f_{n+1}(au)&=f_n(u)\phantom{{}+k_n}=f_n(v)\phantom{{}+k_n}=f_{n+1}(av)\text{, or}\\
f_{n+1}(au)&=f_n(u)+k_n=f_n(v)+k_n=f_{n+1}(av)
\end{align*}
for all $a\in\Sigma$.
The implication from (b) to (c) is trivial. For the last implication consider the case $a=0$, then $f_{n+1}(0u)=f_{n+1}(0v)$ implies $f_{n}(u)=f_{n}(v)$ by definition. If $f_{n+1}(1u)=f_{n+1}(1v)$, then either $f_{n+1}(1u)+k_n>k_{n+1}$ and $f_{n+1}(1u)=f_{n}(u)=f_{n}(v)=f_{n+1}(1v)$ or $f_{n+1}(1u)+k_n\leq k_{n+1}$, which again implies $f_{n}(u)=f_{n}(v)$.

The property \ref{enu:P2} is clear from the definition of $\varphi_n$.
If we assume an additional property on $k_n$, then we are able to bound $\ind{\Theta_n(L)}$ for large $n$.
\begin{lemma}\label{lem:surjectiveBounds}
If $k_n$ grows sufficiently fast, i.e., there is an $N\in\N$ such that $n\cdot k_n\leq k_{2^n}$ for all $n\geq N$,
then %there is an $N\in\N$ such that
\[2^{k_n}\leq \ind\Theta_n(L)\leq 3\cdot(n+1)^2\cdot 2^{k_n}\text{ for all $n\geq 2^N$.}\]
\end{lemma}

\begin{proof}
Let $n\in\N$ with $n\geq 2^N$. First we will determine the number of classes of $\Theta_n(L)$ generated by words of length $2^n$. %in $\Sigma^{2^n}$.
For $u,v\in\Sigma^{2^n}$, if $\varphi_n(u)\not=\varphi_n(v)$, then fix some element $k\in \varphi_n(u)\triangle\varphi_n(v)$. Since $f_n$ is surjective there is a $w\in\Sigma^n$ such that $f_n(w)=k$, and the word $w$ witnesses the fact that $u$ and $v$ are not in the same class. Vice versa, if $\varphi_n(u)=\varphi_n(v)$, then $(u,v)\in\Theta_n(L)$. Now the lower bound immediately follows from \[\ind\Theta_n(L)\geq \left|\{[u]\mid u\in\Sigma^{2^n}\}\right|=|\im \varphi_n| = 2^{k_n}.\]

For the upper bound there is much more work to do.
We will look at the different types of words and bound the number of equivalence classes generated by words of each type separately.

\begin{figure}[ht]
\centering
\def\boxSize{7pt}
\def\patI{north east lines}
\def\patII{crosshatch}
\def\patIII{north west lines}
\def\patIV{crosshatch dots}

\begin{tikzpicture}[scale=1.15]
\def\xl{2.5}
\def\xn{5.2}
\def\xk{8.2}
%vertical lines
\draw (0,3pt)--(0,-3pt) node [below] {\tiny 0};
\draw (1.5,3pt)--(1.5,-3pt) node [below] {\tiny $n$};

\draw (\xl,3pt)--(\xl,-3pt) node [below] {\tiny $2^l+l-n$};
\draw (\xl+0.9,3pt)--(\xl+0.9,-3pt) node [below] {\tiny $2^l$};
\draw (\xl+1.5,3pt)--(\xl+1.5,-3pt) node [below] {\tiny $2^l+l$};

\draw (\xn,3pt)--(\xn,-3pt) node [below] {\tiny $2^n$};
\draw (\xn+1.5,3pt)--(\xn+1.5,-3pt) node [below] {\tiny $2^n+n$};

\draw (\xk,3pt)--(\xk,-3pt) node [below] {\tiny $2^k$};
\draw (\xk+1,3pt)--(\xk+1,-3pt) node [below] {\tiny $2^k+k-n$};
\draw (\xk+2.5,3pt)--(\xk+2.5,-3pt) node [below] {\tiny $2^k+k$};

%\horizontal lines
\draw[thick] (0,0) -- (1.5,0);
\draw[thick, dashed] (1.5,0) -- (\xl-0.35,0);
\draw[thick] (\xl-0.3,0) -- (\xl+1.8,0);
\draw[thick, dashed] (\xl+1.85,0) -- (\xn-0.35,0);
\draw[thick] (\xn-0.35,0)-- (\xn+1.8,0);
\draw[thick, dashed] (\xn+1.85,0) -- (\xk-0.35,0);
\draw[thick] (\xk-0.35,0)-- (\xk+2.8,0);

%regions
\draw[pattern=\patII] (0,0) rectangle (1.5,\boxSize);

%\fill[pattern=\patIV] (1.5,0) rectangle (1.8,\boxSize);
%\draw (1.5,\boxSize) -- (1.8,\boxSize);
\fill[pattern=\patIV] (\xl-0.3,0) rectangle (\xl,\boxSize);
\draw (\xl-0.3,\boxSize) -- (\xl,\boxSize);
\fill[pattern=\patIV] (\xl+1.5,0) rectangle (\xl+1.8,\boxSize);
\draw (\xl+1.5,\boxSize) -- (\xl+1.8,\boxSize);
\fill[pattern=\patIV] (\xn-0.3,0) rectangle (\xn,\boxSize);
\draw (\xn-0.3,\boxSize) -- (\xn,\boxSize);
\fill[pattern=\patIV] (\xn+1.5,0) rectangle (\xn+1.8,\boxSize);
\draw (\xn+1.5,\boxSize) -- (\xn+1.8,\boxSize);
\fill[pattern=\patIV] (\xk-0.3,0) rectangle (\xk+1,\boxSize);
\draw (\xk-0.3,\boxSize) -- (\xk+1,\boxSize);
\fill[pattern=\patIV] (\xk+2.5,0) rectangle (\xk+2.8,\boxSize);
\draw (\xk+2.5,\boxSize) -- (\xk+2.8,\boxSize);

\draw[pattern=\patIII] (\xl,0) rectangle (\xl+0.9,\boxSize);

\draw[pattern=\patI] (\xl+0.9,0) rectangle (\xl+1.5,\boxSize);
\draw[pattern=\patI] (\xn,0) rectangle (\xn+1.5,\boxSize);
\draw[pattern=\patI] (\xk+1,0) rectangle (\xk+2.5,\boxSize);

%legend
\end{tikzpicture}

\begin{tikzpicture}[scale=1.15]
\draw[pattern=\patI] (0,0) rectangle  node[right,xshift=4pt] {I} (\boxSize,\boxSize);
\draw[pattern=\patIII] (3cm,0) rectangle  node[right,xshift=4pt] {II} (3cm+\boxSize,\boxSize);
\draw[pattern=\patII] (6cm,0) rectangle  node[right,xshift=4pt] {III} (6cm+\boxSize,\boxSize);
\draw[pattern=\patIV] (9cm,0) rectangle  node[right,xshift=4pt] {IV} (9cm+\boxSize,\boxSize);

\end{tikzpicture}
\caption{The four different types of words}
\end{figure}
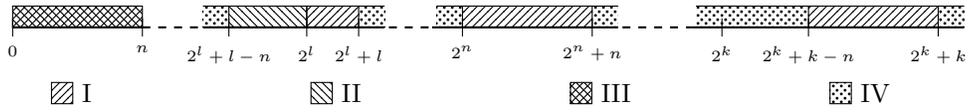

\begin{itemize}
\item Words of type I are of the form $uv$ where $u\in\Sigma^{2^k}$ and $v\in\Sigma^{(k)}$ for some $k\geq \log_2n$. All words $v'\in\Sigma^{(n)}$ for which $uvv'\in L$ have to be of length $k-|v|$. Recall $uvv'\in L$ iff $f_k(vv')\in\varphi_k(u)$. Hence the witnesses can only be used to determine $\varphi_k(u)$. Note that the decomposition into $u$ and $v$ is unique.

\item Every word $u$ of type II is in $\Sigma^{2^l-k}$ for some $l$ and $k$ with $\log_2n\leq l\leq n$ and $k+l\leq n$. All words $w\in\Sigma^{(n)}$ %\FlorianSagt{positive witness if this term is ok} 
with $uw\in L$ are of the form $u'v$ with $u'\in\Sigma^k$ and $v\in\Sigma^l$. The difference to words of type I is that now $\varphi_k(uu')$ depends on the choice of the witness. Note that this could potentially lead to a lot of new equivalence classes.

\item The words of type III are all words with length at most $n-1$. Some of these words $u$	 are short enough such that there can be positive witnesses of $u$ in $\Sigma^{(n)}$ with different lengths. %$w\in\Sigma^{(n)}$ of $u$ with different lengths. % with $uw\in L$. %Note that there are only $2^n-1$ of these words, so they could only influence the entropy of $L$ if $k_n\leq n$. Otherwise ...

\item Finally, for all words $u$ of type IV we have $uw\notin L$ for all $w\in\Sigma^{(n)}$. Hence these words are all in the same equivalence class.
\end{itemize}

We denote number of equivalence classes the words of type I generate, i.e., 
$|\{[u]\mid u\text{ is of type I}\}|$,
by $\bound{I}$. Analogously we define $\bound{II}$, $\bound{III}$, and $\bound{IV}$. Since any word in $\Sigma^*$ is of one of the four types we have that
\[\ind\Theta_n(L)\leq \bound{I}+\bound{II}+\bound{III}+\bound{IV}.\]
We have already noted that $\bound{IV}=1$.
Before we determine an upper bound for $\bound{I}$, $\bound{II}$, and $\bound{III}$ we will look at sets of the form $\Sigma^{2^n}\cdot\{0\}^k$ and determine an upper bound for the size of
$\{[u]\mid u\in \Sigma^{ 2^n}\cdot0^k\}$
for $k\in\{0,\dots,n\}$.

Let $u_1,u_2\in\Sigma^{ 2^n}\cdot0^k$. Then $u_1=u'_1v$, $u_2=u'_2v$ for some $u'_1,u'_2\in\Sigma^{2^n}$ and $v=0^k$.
If $(u_1,u_2)\notin\Theta_n(L)$, then there is some $w\in\Sigma^{n-k}$ which bares witness to this fact. Now $vw$ witnesses $(u'_1,u'_2)\notin\Theta_n(L)$. Hence the words in $\Sigma^{ 2^n}\cdot0^k$ give rise to at most as many classes as the words in $\Sigma^{2^n}$, which, as we have already shown, decompose into $2^{k_n}$ classes.
Therefore we conclude 
\begin{align*}
\left|\{[u]\mid u\in \Sigma^{ 2^n}\cdot0^k\}\right|\leq 2^{k_n}\text{ for all $n\in\N$, $k\in\{0,\dots,n\}$}.\tag{$*$}
\end{align*}

%Recall the definition of $U_w$. For $w\in\Sigma^*$ we denote $\{v\in\Sigma^{(n)}\mid wv\in L\}$, the set of positive witnesses of $w$ of length at most $n$, by $U_w$.  Note that $U_w$ is just a different representation of $[w]$, i.e., $[w_1]=[w_2]$ iff $U_{w_1}=U_{w_2}$.

To find an upper bound for $\bound{I}$ we will show that every word of type~I is already in the same class as some word in $\Sigma^{2^n}\cdot\{0\}^{(n)}$. Let $uv$ be a word of type I with $u\in\Sigma^{2^k}$ and $v\in\Sigma^{(k)}$. We have already noticed that then $\Pos^{(n)}_{uv}\subseteq \Sigma^{k-|v|}$. Define $v'=0^{n-k+|v|}$. Since $n-k+|v|+k-|v|=n$ and $\varphi_n$ is surjective there is an $u'$ such that
$\varphi_n(u')=f_n(v' \Pos^{(n)}_{uv})$.
Note that $u'v'\in\Sigma^{2^n}\cdot\{0\}^{(n)}$. 
Our next goal is to show that $uv$ and $u'v'$ are in the same equivalence class, i.e., \[w\in \Pos^{(n)}_{uv} \iff f_n(v'w)\in \varphi_n(u')\stackrel{\mathrm{Def.}}{\iff}w\in \Pos^{(n)}_{u'v'} .\]
The $\Rightarrow$ direction is straightforward. If $w\in \Pos^{(n)}_{uv}$, then $v'w\in v'\Pos^{(n)}_{uv}$ and therefore $f_n(v'w)\in f_n(v'\Pos^{(n)}_{uv})=\varphi_n(u')$.
For the $\Leftarrow$ direction take some word $w$ with $f_n(v'w)\in\varphi_n(u')$. By choice of $u'$ there is a $w'\in \Pos^{(n)}_{uv}$ such that $f_n(v'w)=f_n(v'w')$.
Then
\begin{align*}
f_n(v'w)=f_n(v'w') &{}\iff f_{k-|v|}(w)=f_{k-|v|}(w')\tag{$|v'|$ times \ref{enu:P1}}\\
&\iff f_k(vw)=f_k(vw').\tag{$|v|$ times \ref{enu:P1}}
\end{align*}
Because $w'\in \Pos^{(n)}_{uv}$ implies $f_k(vw)=f_k(vw')\in\varphi_k(u)$ we conclude $w\in \Pos^{(n)}_{uv}$.
Thus $uv$ is in the same class as $u'v'$.

Consequently, we can finally give an upper bound for $\bound{I}$:
\begin{align*}
\bound{I}\leq \left|\{[w]\mid w\in \Sigma^{2^n}\cdot\{0\}^{(n)}\}\right|
\leq \sum_{k=0}^n \left|\{[w]\mid w\in \Sigma^{2^n}\cdot0^{k}\}\right|%\underbrace{|\{[w]\mid w\in \Sigma^{2^n}\cdot0^{k}\}|}_{\leq 2^{k_n}}
\stackrel{(*)}\leq(n+1)\cdot2^{k_n}.
\end{align*}

%TODO find a good picture to explain this
For $\bound{II}$ we will consider all words of type II with the same length $2^l-k$ for some $l,k\in\N$ with $l+k\leq n$. For fixed $l$ and $k$ we will bound the size of $\{\Pos^{(n)}_u \mid u\in \Sigma^{2^l-k}\}$ by $2^{k_l}$. To do this let $u'\in\Sigma^k$. We show that
\begin{align*}
\{\Pos^{(n)}_{u} \mid u\in\Sigma^{2^l-k}\} &\to \im\varphi_l\\
\Pos^{(n)}_{u}&\mapsto \varphi_l(uu')
\end{align*}
is a well defined injective map. Firstly, we tackle the problem of well definedness. Let $u_1,u_2\in\Sigma^{2^l-k}$ such that $\Pos^{(n)}_{u_1}=\Pos^{(n)}_{u_2}$. Then
\begin{align*}
f_l(v )\in\varphi_l(u_1u')\iff u'v \in \Pos^{(n)}_{u_1}
\iff u'v \in \Pos^{(n)}_{u_2}
\iff f_l(v)\in\varphi_l(u_2u')
\end{align*}
and $\varphi_l(u_1u')=\varphi_l(u_2u')$.
For injectivity take two words $u_1,u_2\in \Sigma^{2^l-k}$ with $\varphi_l(u_1u')=\varphi_l(u_2u')$. Let $w\in \Pos^{(n)}_{u_1}$. Then decompose $w$ into $u''v'$ with $u''\in\Sigma^{k}$ and $v'\in\Sigma^l$. Now, by definition, $v'\in\varphi_l(u_1u'')$. Furthermore, from \ref{enu:P2} we can deduce $\varphi_l(u_1u'')=\varphi_l(u_2u'')$, and thus $v'\in\varphi_l(u_2u'')$. Because of this $w\in \Pos^{(n)}_{u_2}$ and $\Pos^{(n)}_{u_1}\subseteq \Pos^{(n)}_{u_2}$. By exchanging the roles of $u_1$ and $u_2$ equality of $\Pos^{(n)}_{u_1}$ and $\Pos^{(n)}_{u_2}$ follows. Therefore
\[\left|\{[u] \mid u\in \Sigma^{2^l-k}\}\right|=\left|\{\Pos^{(n)}_{u} \mid u\in\Sigma^{2^l-k}\}\right|\leq |\im\varphi_l|= 2^{k_l}\leq 2^{k_n}.\]
Recall that $l+k\leq n$. Consequently,
\[\bound{II}\leq (n+1)^2\cdot 2^{k_n}.\]

To bound $\bound{III}$ we use the same method we used for $\bound{II}$. Fix a $k<n$ and consider all classes generated by words in $\Sigma^k$.
Firstly, note that for an $u\in\Sigma^k$ there can be words of different lengths in $\Pos^{(n)}_u$, but we still know that $\Pos^{(n)}_u\subseteq \bigcup_{l\leq m} \Sigma^{2^l+l-k}$ where $m=\max \{l\in\N\mid 2^l + l-k\leq n\}$ and $\Sigma^{-i}=\emptyset$. 
Furthermore, $m\leq\lceil\log_2n\rceil$, because
\[2^{\lceil\log_2n\rceil+1}+\lceil\log_2n\rceil+1-k\geq 2\cdot n+1-(n-1)> n.\]

But $\Pos^{2^l+l-k}_u$ contains by definition only witnesses of the same length, hence we can apply the same argument as before to %$\Pos^{2^l+l-k}_u$ %Denote $\Pos_u\cap \Sigma^{2^l+l-k}$ by $\Pos_{u,l}$. To these sets we can then apply the same argument as before 
obtain $|\{\Pos^{2^l+l-k}_{u} \mid u\in\Sigma^{k}\}|\leq |\im\varphi_l|\leq 2^{k_l}$.
Since $k_0=1$ we deduce
\begin{align*}
\left|\{\Pos^{(n)}_u \mid u\in\Sigma^{k}\}\right|\leq\prod_{l\leq m}\left|\{\Pos^{2^l+l-k}_{u} \mid u\in\Sigma^{k}\}\right|
\leq\prod_{l\leq m} 2^{k_l}
\leq 2\cdot\prod_{l=1}^m 2^{k_l}%\tag{$k_0=1$}
\leq 2\cdot 2^{m\cdot k_m}.%\tag{$k_l\leq k_m$}.
\end{align*}

We know that $m\leq \log_2n$, also $n\geq 2^N$ implies $\log_2n\geq N$. Furthermore, if we assume for the sake of readability that $\log_2n$ is a natural number, then by assumtion that $k_n$ grows sufficiently fast we have %TODO log_2n natural number is ok because of Lemma \ref{lem:limsupMonotone}
\begin{align*}
m\cdot k_m&\leq \log_2n\cdot k_{\log_2n}\leq k_n.
\end{align*}
As there are $n$ possible values for $k$ we obtain $\bound{III}\leq 2\cdot n\cdot 2^{k_n}$.

Now we can finally give an upper bound for $\ind\Theta_n(L)$:
\begin{align*}
\ind\Theta_n(L)&\leq \bound{I}+\bound{II}+\bound{III}+\bound{IV}\\
&\leq (n+1)\cdot 2^{k_n}+(n+1)^2\cdot 2^{k_n}+2\cdot n\cdot 2^{k_n}+1\\
&\leq 3\cdot(n+1)^2\cdot 2^{k_n}.
\end{align*}
This finishes the proof.
\end{proof}
Note that this lemma can be applied to any surjective functions $\varphi_n$ and $f_n$ fulfilling the properties \ref{enu:P1} and \ref{enu:P2}, not just the $\varphi_n$ and $f_n$ we defined.
Now we can, using this lemma, show the main result from this section.
\begin{theorem}\label{the:entropyIsSur}
For any alphabet $\Sigma$ with at least two letters the entropy function
\begin{align*}
\eta\colon \mathcal{P}(\Sigma^*)&\to[0,\infty]\\
L&\mapsto \eta(L)
\end{align*}
is surjective.
\end{theorem}
\begin{proof}
For 0 and $\infty$ we have already seen that there are languages with that entropy.
Thus let $x$ be a positive real number. We would like to use the sequence $k_n=\lceil n\cdot x\rceil$, but it is not suitable, because $k_0=0$. Furthermore, $k_1$ could be larger than $2$. Because of this we define $k_n=\max \{\min \{\lceil n\cdot x\rceil,2^n\},1\}$. Now the sequence $(k_n)_{n\in\N}$ is suitable. %Note that $k_0=1$ and $k_{n+1}\leq 2\cdot k_{n}$. 
Since $\Sigma$ has at least two letters we can define $f_n$, $\varphi_n$, and $L$ as above. Clearly, there is an $N_1\in\N$ such that $k_n=\lceil n\cdot x\rceil$ for all $n\geq N_1$. Furthermore, there is an $N_2\in\N$ such that for all $n\geq \max\{N_1,N_2\}$
\begin{align*}
n\cdot k_n=n\cdot \lceil n\cdot x\rceil\leq n\cdot (n\cdot x+1)=n\cdot(n+\frac{1}{x})\cdot x\leq 2^n\cdot x\leq k_{2^n}.
\end{align*}

Hence we apply Lemma \ref{lem:surjectiveBounds} with $N=\max\{N_1,N_2\}$ to obtain
\[2^{k_n}\leq \ind\Theta_n(L)\leq 3\cdot(n+1)^2\cdot 2^{k_n}\]
for all $n\geq 2^N$.
Now we easily compute
\begin{align*}
\eta(L)\geq \limsup_{n\to\infty}\frac{\log_2(2^{k_n})}{n}=\limsup_{n\to\infty}\frac{\lceil n\cdot x\rceil}{n}=x
\end{align*}
and
\begin{align*}
\eta(L)\leq \limsup_{n\to\infty}\frac{\log_2(3\cdot(n+1)^2\cdot 2^{k_n})}{n}=x.
\end{align*}

Therefore, $\eta(L)=x$, and we conclude that $\eta$ is surjective.
\end{proof}

Naturally the question arises whether same holds for unary alphabets.
In the next subsection we will address this question.

\subsection{Unary Languages}\label{sec:UnaryLanguages} %alone we are weak .... but together we are strong

For the remainder of this section let $\Sigma=\{a\}$ be a unary alphabet. 
We shall write $n$ instead of $a^n$. %Throughout this section we will write $n$ instead of $a^n$. 
Then we can view $\Sigma^*$ as $\N$ and a unary language $L$ is just a subset of $\N$. %TODO note that the topologicical automata are now time-diskret dynamical systems --> maybe some results from DYSYS can be applied
We will show that the entropy of a language over a unary alphabet can be bounded by one and we will show that this upper bound is tight. 
%Before we continue let us look at some examples of unary languages.
%We have already seen some examples in Chapter~\ref{sec:Languages}, namely the set of all square numbers $\{n^2\mid n\in\N\}$ and the set of all prime numbers $\{p\in\N\mid\text{$p$ is prime}\}$.
%We have already seen a unary language in Example \ref{exa:n2} and in fact many interesting languages are unary. For every function $f\colon \N\to\N$ we can define $L=\im f$ and if $f$ is strictly monotone increasing then we can retrieve $f$ from $L$. Interesting cases for $f$ are $n^2$, $2^n$, and $2^{2^n}$. %TODO maybe find some really interesting examples ;)
%Another unary language is the set of all prime numbers from Example \ref{exa:prime}.
%Both of these examples are decidable and not very complicated. 
Note that there are undeciable unary languages, for example we can encoded the halting problem using an enumeration $M_1,M_2,\dots$ of all Turing machines and define $L=\{i\in\N\mid M_i\text{ halts on input $\varepsilon$}\}$.
This makes the following result a bit surprising.

%\subsection*{Maximal entropy of unary languages}
\begin{theorem}\label{the:unaryMax}
Let $L$ be a unary language. Then $\eta(L)\leq1$.
\end{theorem}
\begin{proof} 
Let $L\subseteq\N$ be a unary language. Clearly, $\ind\Theta_n(L)\leq 2^{n+1}$ and therefore
\begin{align*}
\eta(L)&=\limsup_{n\to\infty}\frac{\log_2 (\ind \Theta_n(L))}{n}\leq \limsup_{n\to\infty}\frac{\log_2 (2^{n+1})}{n}=1.
\end{align*}
This finishes the proof.
\end{proof}

Using a construction similar to the one in the previous subsection we can show that this bound is tight.
\begin{example}
For any $n\in\N$ let $\varphi_n\colon\{0,\dots,2^{n+1}-1\}\to \mathcal{P}(\{0,\dots,n\})$ be a bijection.
Note that if we know that a number is of the form $2^n+k$ and $k< 2^n$, then we can uniquely determine $k$ and $n$ from that number. With this in mind we define:
\[\UInfty=\{2^{n+2^{n+m}}+k\mid n\in\N, m\leq 2^{n+1}-1, k\in\varphi_n(m)\}.\]

Observe that from any number of the form $2^{n+2^{n+m}}+k$ uniquely determines the numbers $n$, $m$, and $k$.
%To decompose any given natural number $n_1$ uniquely into $2^{n+2^{n+m}}+k$, find $n_2$, the largest power of 2 less or equal than $n_1$. Then $n_1=2^{n_2}+k$ for some $k<2^{n_2}$. Use the same method to decompose $n_2=2^{n_3}+n$ with $n<2^{n_3}$ and set $m=n_3-n$. Note that $m$ can be negative and $k$ can be greater than $n$. For example, the unique decomposition of 1025 is $2^{10}+1=2^{2+2^{2+1}}+1$, so $n=2$, $m=1$, and $k=1$ and $1025\in\UInfty$ if $1\in\varphi_2(1)$. The reason we need $n+m$ in the exponent instead of just $m$ is that the value $2^m$ is for small $m$ not necessarily larger than $n$, which would make the decomposition impossible. Note that $\lceil\log_2(n+1)\rceil+m$ would have been sufficient but for readability we take $n+m$ instead. 
To compute the entropy of $\UInfty$ let $n\in\N$. Firstly, we denote $2^{n+2^{n+m}}$ by $w_{n,m}$. Now let us consider the set $\{w_{n,m}\mid m\leq 2^{n+1}-1\}$. If $w_{n,m}\neq w_{n,m'}$, then $m\neq m'$ and since $\varphi_n$ is injective we have that there is a $k\in \varphi_n(m)\triangle\varphi_n(m')$. This $k$ witnesses that $(w_{n,m},w_{n,m'})\notin\Theta_n(\UInfty)$. Hence 
\[\ind\Theta_n(\UInfty)\geq |\{w_{n,m}\mid m\leq 2^{n+1}-1\}|=2^{n+1}.\]
%Therefore
%\begin{align*}
%\eta(\UInfty)&\geq\limsup_{n\to\infty}\frac{\log_2 2^{n+1}}{n}=1.
%\end{align*}
Together with Theorem $\ref{the:unaryMax}$ we can conclude that $\eta(\UInfty)=1$.
\end{example}

%We have seen in the previous section that the entropy function is surjective over every alphabet with at least to letters. A natural next question to ask is whether there is a similar statement for unary alphabets.
We have just seen that the entropy function is surjective over every alphabet with at least two elements. Hence we conjecture that the following holds.
\begin{conjecture}\label{con:surjUnary}
For any unary alphabet $\Sigma$ we have that the entropy function 
\begin{align*}
\eta\colon \mathcal{P}(\Sigma^*)&\to[0,1]\\
L&\mapsto \eta(L)
\end{align*}
 is surjective.
\end{conjecture}
Unfortunately, we were not able to prove this conjecture. In the next section we will look at the entropy of decision problems and use an observation there to strengthen the surjectivity result for at least two element alphabets.

\section{Entropy of Decision Problems}\label{sec:DecisionProblems}

%We have just seen a connection between topological entropy and the Chomsky hierarchy. 
In this chapter we will connect topological entropy with decision problems. First we will compute the entropy of the \textbf{NP}-complete problem \textbf{SAT}. Then we will use padding to show that the entropy of any language can be reduced to zero. In particular this shows that there are undecidable languages with zero entropy.

To be able to compute the entropy of \textbf{SAT}, we need to define a suitable encoding. We use the alphabet $\{(,),\wedge,\vee,\neg,0,1\}$. % and define the encoding of a formula $\varphi$, denoted by $\langle \varphi\rangle$, as 
To encode a formula $\varphi$ we replace every variable $x_n$ by $\operatorname{bin}(n)$, the binary representation of $n$.
We denote the encoded formula by $\langle\varphi\rangle$. For example $\langle x_1\wedge x_2\rangle=1\wedge10$.
Now we can define
\[\textbf{SAT}=\{\langle\varphi\rangle\mid \text{$\varphi$ is satisfiable}\}.\]

%\begin{align*}
%\langle x_n\rangle&=\operatorname{bin}(n)&\langle \varphi_1\wedge\varphi_2\rangle&=(\langle \varphi_1\rangle\wedge\langle \varphi_2\rangle)\\
%\langle \neg\varphi\rangle&=\neg \langle \varphi\rangle & \langle \varphi_1\vee\varphi_2\rangle&=(\langle \varphi_1\rangle\vee\langle \varphi_2\rangle)\\
%\end{align}

\begin{lemma}
The language \textbf{SAT} has infinite entropy.
\end{lemma}
\begin{proof}
Consider the set 
\[\{\langle L_1\wedge\dots\wedge L_{2^n}\rangle\mid L_1\in\{x_1,\neg x_1\},\dots, L_{2^n}\in\{x_{2^n},\neg x_{2^n}\}\}.\]
Take two words $w_1=\langle\varphi_1\rangle$ and $w_2=\langle\varphi_2\rangle$ from this set. If $w_1\not=w_2$, then there is some $k\in\{1,\dots,2^n\}$ such that the $k$\textsuperscript{th} literal of $\varphi_1$ and $\varphi_2$ differ. Without loss of generality assume that the $k$\textsuperscript{th} literal of $\varphi_1$ is $x_k$ and the $k$\textsuperscript{th} literal of $\varphi_2$ is $\neg x_k$. Then $\varphi_1\wedge x_k$ is satisfiable and $\varphi_2\wedge x_k$ is not.
Note that $|\wedge\operatorname{bin}(k)|\leq 1+\log_2{2^n}=1+n$. Therefore $\wedge\operatorname{bin}(k)$ witnesses $(w_1,w_2)\notin\Theta_{n+1}(\textbf{SAT})$.
Since the set contains $2^{2^n}$ words we can now show that infinity is a lower bound for the entropy of \textbf{SAT}
\begin{align*}
\eta(\textbf{SAT})&\geq\limsup_{n\to\infty}\frac{\log_2(2^{2^n})}{n+1}=\infty.
\end{align*}
This concludes the proof.
%For \textbf{SAT} take $\wedge x_k$ as witness for $k\leq 2^n$. %maybe also \wedge \neg x_k
\end{proof}
%Note that \textbf{$\boldsymbol 2$COLORING} is in \textbf{P} and the same construction we used for \textbf{SAT} could also be used for \textbf{$\boldsymbol 2$SAT}. Hence there does not seem to be an obvious distinction between \textbf{NP}-complete problems and problems in \textbf{P} from the point of view of topological entropy. 

Next we will discuss the effect padding has on the complexity of a language. In complexity theory padding can be used to decrease the complexity of a language. What happens for topological entropy? For a language $L$ over $\Sigma$ define
\[\PAD(L)=\{uv\mid u\in L, v\in \Sigma^{2^{{|u|}}}\}.\]
Note that if $L$ is in \textbf{EXPTIME}, then $\PAD(L)$ is in \textbf{P}. So $\PAD(L)$ is much easier than $L$, and this decrease in complexity is also reflected in the topological entropy of $\PAD(L)$.

\begin{theorem}
Let $L$ be a language over $\Sigma$ with $|\Sigma|\geq 2$. Then
\[\eta(\PAD(L))=0.\]
\end{theorem}
\begin{proof}
Note that every word in $\PAD(L)$ has a length of the form $2^{k}+k$. 
Consider the set $\{\Pos^{(n)}_u\mid |u|\geq \log_2n\}$. We will show that every $\Pos^{(n)}_u$ is either the empty set or $\Sigma^k$ for some $k\in\{0,1\dots,n\}$. Let $u\in\Sigma^*$ with $|u|\geq \log_2n$.
For every $k\geq\log_2n $ we have: 
\begin{align*}
(2^{{k+1}}+(k+1))-({2^k}+k)&={2\cdot2^k}-{2^{k}}+1={2^k}+1\geq n+1. 
\end{align*}
Hence the lengths of words in $\PAD(L)$ are so far apart that $\Pos^{(n)}_u\subseteq \Sigma^k$ for some $k\in\{0,1\dots,n\}$. If $\Pos^{(n)}_u$ is not empty, then there is some $w\in \Pos^{(n)}_u$ and $uw\in\PAD(L)$.
We know that $uw$ is of the form $u'v'$ for some $u'\in L$ and $v'\in\Sigma^*$ with $|v'|={2^{|u'|}}$. 
Since $2^{|u|}\geq n\geq |w|$ we have that $w$ is a postfix of $v'$.

By definition $u'v''\in\PAD(L)$ for all $v''$ with $|v''|=2^{{|u'|}}$. As a consequence, $uv''\in\PAD(L)$ for all $v''\in\Sigma^k$ and therefore $\Pos^{(n)}_u=\Sigma^k$.
Hence we can bound the number of classes in $\Theta_n(\PAD(L))$ by
\begin{align*}
\ind\Theta_n(\PAD(L))&\leq \left|\{\Pos^{(n)}_u\mid |u|<\log_2n\}\right|+\left|\{\Pos^{(n)}_u\mid |u|\geq \log_2n\}\right|\\
&\leq \left|\Sigma^{(\log_2n)}\right|+n+2.
%&\leq |\Sigma|^{\log_2n+1}+n+2.\tag{Lemma \ref{lem:geometric}}
\end{align*}
Now we can determine the entropy
\begin{align*}
\eta(\PAD(L))\leq \limsup_{n\to\infty}\frac{\log_2(|\Sigma^{\log_2n}|+n+2)}{n}
%&\leq\limsup_{n\to\infty}\frac{(\log_2n+1)\cdot\log_2|\Sigma|+\log_2n+\log_22}{n}\\
=0.
\end{align*}
This finishes the proof.
\end{proof}
Note that this works for any language $L$, even if it is undecidable, and since $L$ can be reconstructed from $\PAD(L)$ we know that $\PAD(L)$ is also undecidable.
\begin{corollary}
There are undecidable languages with zero entropy. 
\end{corollary}
This is a rather strange result, because the barrier of 
undecidability cannot be breached in classical complexity theory. Undecidable languages are always complicated. %We can embrace this or change the definition.

We will use this result to show that not only is the entropy function surjective, there are even  uncountably many languages for every entropy.
\begin{corollary}\label{cor:uncountZero}
Over an at least two element alphabet $\Sigma$ there are uncountably many languages with zero entropy.
\end{corollary}

\begin{proof}
There are uncountably many languages over $\Sigma$ and as mentioned before $\PAD$ is injective and every language in its image has zero entropy.
\end{proof}

The following lemma will help us to construct uncountably many languages for every entropy, not just for zero.

\begin{lemma}\label{lem:entropyOfCombLangIsMax}
Let $L_1,L_2\subseteq \Sigma^*$ be nonempty languages and $\texttt{\#}$ a new symbol not in $\Sigma$. Then 
\[\eta(L_1\texttt{\#}L_2)=\max\{\eta(L_1),\eta(L_2)\}.\]
\end{lemma}
\begin{proof}
Let $v'\in L_2$ with $|v'|$ minimal. We show that for $n\geq|v'|$ 
\[\ind\Theta_n(L_1\texttt{\#}L_2)= \ind\Theta_{n-|v'|-1}(L_1)+\ind\Theta_{n}(L_2)+k,\]
for $k$ either 1 or 0. Consider the map
\begin{align*}
(\Sigma^*\texttt{\#}\Sigma^*)/\Theta_n(L_1\texttt{\#}L_2)&\to \Sigma^*/\Theta_{n-|v'|-1}(L_1)\sqcup \Sigma^*/\Theta_n(L_2)\\
[u]^n &\mapsto [u]^{n-|v'|-1}_1\\
[u\texttt{\#}v]^n &\mapsto [v]^n_2,
%[\texttt{\#}\texttt{\#}]&\mapsto%only map from sigma^*#sigma^* -> sigma^* cup sigma^*
\end{align*}
where the sets in the image are potentially renamed to make the classes for $L_1$ and $L_2$ disjoint. It is left as an exercise to the reader to show that this map is well defined. It is clearly surjective and since $v'$ is of minimal length and $[u\texttt{\#}v]^n=[u'\texttt{\#}v]^n$ also injective. Note that any word in $(\Sigma\cup\{\texttt{\#}\})^*\setminus(\Sigma^*\texttt{\#}\Sigma^*)$ has no positive witnesses and is in the same class as $\texttt{\#}\texttt{\#}$. This shows the above equality. Therefore 
\begin{align*}
\eta(L_1\texttt{\#}L_2)&=\limsup_{n\to\infty}\frac{\ind\Theta_{n-|v'|-1}(L_1)+\ind\Theta_{n}(L_2)}{n}\\
&=\max\left(\limsup_{n\to\infty}\frac{\ind\Theta_{n}(L_1)}{n},\limsup_{n\to\infty}\frac{\ind\Theta_{n}(L_2)}{n}\right)=\max\{\eta(L_1),\eta(L_2)\}
\end{align*}
as desired.

%First we show that $\eta(L_1\texttt{\#}L_2)\geq \eta(L_1)$. Let $v\in L_2$. If we consider witnesses of the form $u\texttt{\#}v$ it is clear that $\ind\Theta_n(L_1\texttt{\#}L_2)\geq \ind\Theta_{n-|v|-1}(L_1)$.

%Similarly, we show that $\eta(L_1\texttt{\#}L_2)\geq \eta(L_2)$. Let $u\in L_1$ and consider words of the form $u\texttt{\#}v$. These words give rise to $\ind\Theta_n(L_2)$ classes in $\Theta_n(L_1\texttt{\#}L_2)$.

%Finally, we show that $\eta(L_1\texttt{\#}L_2)\leq\max\{\eta(L_1),\eta(L_2)\}$

%\[\ind\Theta_n(L_1\texttt{\#}L_2)= \ind\Theta_{n-|v|-1}(L_1)+\ind\Theta_{n}(L_2) \text{ if $v$ is of minimal length? there can be +-1 if no language has a garbage class }\]

%\begin{align*}
%[u]^n &\mapsto [u]^{n-|w|-1}_1\\
%[u\texttt{\#}v]^n &\mapsto [v]^n_2\\
%[\texttt{\#}\texttt{\#}]&\mapsto%only map from sigma^*#sigma^* -> sigma^* cup sigma^*
%\end{align*}
%bijective and well defined

\end{proof}

We can use this to strengthen Theorem \ref{the:entropyIsSur}.

\begin{corollary}
For any alphabet $\Sigma$ with at least two letters and any $x\in[0,\infty]$ the set
\[\{L\subseteq \Sigma^*\mid \eta(L)=x\}\]
is uncountable.
\end{corollary}
\begin{proof}
Assume that $\Sigma$ has just two elements. By Corollary \ref{cor:uncountZero} the claim holds for $x=0$. Let $x\in (0,\infty]$. By Theorem \ref{the:entropyIsSur} there exists a language $L$ with entropy $2x$. Now, by Lemma \ref{lem:entropyOfCombLangIsMax}, every language in the uncountable set $\{L_0\texttt\#L\mid \eta(L_0)=0\}$ has entropy $2x$. Note that we have introduced a new symbol, hence we encode every language from the set over $\Sigma$ and obtain by \ref{cor:encBinary} an uncountable set of languages with entropy $x$.

\end{proof}

%leftpadding decreases complexity class but leaves entropy unchanged (does it? $\Sigma^{(n)}\to\Sigma^n$ sieht nicht gut aus)

%also right padding here, because of undecidability -> 0 entropy

\textbf{}\section{Topological Entropy and Automata}\label{cha:automata}
Now we come to the second part of this article. Here we introduce $k$-stack push-down automaton ($k$-stack PDA) and $k$-counter automata. We show that the entropy of a language that is accepted by a $k$-stack PDA can be bounded in terms of the number of stack symbols. As a corollary we obtain that all languages recognized by $k$-counter automata have zero entropy, proving an open conjecture from Schneider and Borchmann \cite{schneiderborchmann}.
%A $k$-stack push-down automaton ($k$-stack PDA) is a finite automaton equipped with $k$ stacks. 

Before we start with the formal definition of $k$-stack PDAs, we fix some functions and conventions. We will write $\boldsymbol{v}$ for the tuple $(v_1,\dots,v_k)$ from the set  $\Gamma^*_1\times\dots\times\Gamma^*_k$ and $\boldsymbol{\varepsilon}$ for the tuple containing only $\varepsilon$ in each entry.
For a word $u=a_1\dots a_n\in\Sigma^n$ and $I\subseteq\N$ let $\pi_I(u)$ be the \define{projection of $u$ onto $I$}:
\[\pi_I(u)=a_{i_1}\dots a_{i_k} \text{ where $i_1<\dots <i_k$ and $\{i_1,\dots,i_k\}=I\cap\{1,\dots,n\}$}.
	%\begin{cases}
	%u & \text{if }|u|< n\\
	%u_1 & \text{if  $u=u_1u_2$ for some $u_1\in\Sigma^n$ and $u_2\in\Sigma^*$}
	%\end{cases}
\]

Let $\pi_n=\pi_{\{1,\dots,n\}}$, $\head=\pi_{\{1\}}$, and $\tail=\pi_{\{2,3,\dots\}}$. Dually we define the functions $\bottom(u)=\pi_{\{|u|\}}(u)$ and $\front(u)=\pi_{|u|-1}(u)$.
For two words $u,v\in\Sigma^*$ where $v$ is a postfix of $u$ define $u-v$ as $u_1$, where $u=u_1v$. For tuples $\boldsymbol{v}$ the functions $\pi_I$, $\head$, $\tail$, and $-$ are applied componentwise. 
%Note that if $u$ is a word over a unary alphabet, then $\head(u)$ corresponds to $\sgn(|u|)$. 
	
\begin{definition}
A $k$-stack PDA is a tuple $\mathcal{A}=(Q,\Sigma,\Gamma_1,\dots,\Gamma_k, \delta,q_0,F)$ where
\begin{itemize}
\item $Q$ is a finite set (of \define{states}),
\item $\Sigma$ is an alphabet (the \define{input alphabet}),
\item $\Gamma_1,\dots,\Gamma_k$ are alphabets (the \define{stack alphabets}), 
\item $\delta\subseteq \Gamma^{(1)}_1\times\dots\times\Gamma^{(1)}_k\times Q\times (\Sigma\cup\{\varepsilon\})\times Q\times \Gamma^*_1\times\dots\times\Gamma^*_k$ (the \define{transition relation}),
\item $q_0$ is from $Q$ (the \define{initial state}), and
\item $F$ is a subset of $Q$ (the \define{set of final states}).
\end{itemize}

A configuration of $\mathcal{A}$ is a tuple $(q,\boldsymbol{v},w)$ with the current state $q$, the values stored in the stacks $\boldsymbol{v}$, and the remaining input $w$.
For $a\in\Sigma^{(1)}$ we can make the transition
\[(q,\boldsymbol{v},aw)\vdash(p,\boldsymbol{u},w)\text{ if }
(\head(\boldsymbol{v}),q,a,p,\boldsymbol{u}-\tail(\boldsymbol{v}))\in\delta.\] %\DanielSagt{Explizit erwähnen, dass Stacks nach links wachsen}
Beware that a transition of the form $(\boldsymbol \varepsilon,p,a,q,\boldsymbol u)$ can only be used if all stacks are empty. Also the symbol on top of a stack represented by $v$ is the leftmost symbol of $v$. Hence the stacks grow to the left. 
The language accepted by $\mathcal{A}$ is defined as
\[L(\mathcal{A})=\{w\in\Sigma^*\mid \exists p\in F.\  (q_0,\boldsymbol\varepsilon,w)\vdash^*(p,\boldsymbol v,\varepsilon)\}.\]
We call $\mathcal{A}$ \define{deterministic} if for every configuration $K$ there is at most one configuration $K'$ such that $K\vdash K'$, \define{total} if for every configuration $K$ there is at least one configuration $K'$ such that $K\vdash K'$, and \define{$\varepsilon$-free} if there is no transition of the form $(\boldsymbol{v},q,\varepsilon,p,\boldsymbol{u})$ in $\delta$.
\end{definition}

In the following we will always assume that our automata are total.
It is clear that a 0-stack PDA is just an ordinary finite automaton.
We define a \define{$k$-counter automaton} to be a $k$-stack PDA where all stack alphabets are unary. The height of a stack can be seen as the value of a counter. Note that a counter automaton can only test whether a counter is zero or not.
While counter automata do not seem to be that powerful, we have the following surprising result. 
\begin{theorem}[Theorem 7.9 from \cite{hopcrofullman}]\label{the:2counterautomataTuringcomplete}
	A deterministic $2$-counter automaton can simulate an arbitrary Turing machine.
\end{theorem}
Because of this and since topological automata are inherently deterministic we will restrict ourselfs to deterministic $\varepsilon$-free $k$-stack PDAs for now. %Later we will see that also the restriction to deterministic PDAs is necessary to obtain an upper bound for the entropy.

Our next goal is to obtain an upper bound on the entropy of such an automaton. The idea is the following: First, translate the PDA $\mathcal{A}$ into a topological automaton $\Top{A}$. Second, we observe that witnesses of length $n$ can only see the leftmost $n$ symbols on each stack. Third, we use this observation to find an upper bound for $\ind\Lambda_n(\Top A)$ and therefore also for $\eta(L(\mathcal A))$.

Fix a deterministic $\varepsilon$-free $k$-stack PDA $\mathcal{A}=(Q,\Sigma,\Gamma_1,\dots,\Gamma_k, \delta,q_0,F)$.
A naive choice for the states of $\Top{A}$ would be $Q\times\Gamma^*_1\times\dots\times\Gamma^*_k$. 
But recall that the states of a topological automaton are a compact Hausdorff space. Therefore, we adjust this idea a little bit and define for an alphabet $\Gamma$ the set $\Gamma^\infty=\Gamma^*\cup\Gamma^\N$, where $\Gamma^\N$ denotes the set of all (right) infinite words over $\Gamma$. We equip $\Gamma^\infty$ with the topology defined by the following basis of open sets
\[\{\{u\}\mid u\in\Gamma^*\}\cup \{u\Gamma^\infty\mid u\in\Gamma^*\},\]
i.e., the open sets are exactly the sets that can be expressed as unions of sets from the basis. The space $\Gamma^\infty$ is compact. 
\begin{definition}
The topological automaton $\Top{A}$ is defined as \[(Q\times\Gamma^\infty_1\times\dots\times\Gamma^\infty_k,\Sigma,\alpha,(q_0,\boldsymbol\varepsilon),F\times\Gamma^\infty_1\times\dots\times\Gamma^\infty_k),\] where $\alpha((q,\boldsymbol{v}),w)=(p,\boldsymbol{u})$ if $(q,\boldsymbol{v},w)\vdash^*(p,\boldsymbol{u},\varepsilon)$. Here $\vdash^*$ is extended to infinite words in the obvious way. We equip $Q\times\Gamma^\infty_1\times\dots\times\Gamma^\infty_k$ with the product topology, where the topology on $\Gamma_i^\infty$ is the one defined above and the topology on $Q$ is the discrete topology.
\end{definition}
Clearly, we have that 
\[L(\Top A) = L(\mathcal{A}).\]

Now comes the second step of the proof.
\begin{lemma}\label{lem:PDALambda}
For two states $(q,\boldsymbol u)$ and $(p,\boldsymbol v)$ of $\Top{A}$. If $q=p$ and $\pi_n(\boldsymbol u)=\pi_n(\boldsymbol v)$, then $((q,\boldsymbol u),(p,\boldsymbol v))\in\Lambda_n(\Top A)$.
\end{lemma}
\begin{proof}
Let $F'=F\times\Gamma^\infty_1\times\dots\times\Gamma^\infty_k$.
Note that the lemma is equivalent to showing that for any $w\in\Sigma^{(n)}$ we have that $\alpha((q,\boldsymbol u),w)\in F'$ iff  $\alpha((p,\boldsymbol v),w)\in F'$.

The proof is by induction on $n$.

If $n=0$. Then $w=\varepsilon$ and the statement is trivially true.

If $n>0$. Then either $w=\varepsilon$ and the statement is again trivial or $w=aw'$ for some $a\in\Sigma$ and $w'\in\Sigma^{(n-1)}$. Consider $(q',\boldsymbol u')=\alpha((q,\boldsymbol u),a)$ and $(p',\boldsymbol v')=\alpha((p,\boldsymbol v),a)$. Since $n>0$ we have that $\head(\boldsymbol{u})=\head(\boldsymbol{v})$. Furthermore, $p=q$ and therefore both states have to use the same transition. Hence $q'=p'$ and $\pi_{n-1}(\boldsymbol{u}')=\pi_{n-1}(\boldsymbol{v}')$. Applying the induction hypothesis yields
\[\alpha((q,\boldsymbol u),w)=\alpha((q',\boldsymbol u'),w')\in F' \text{ iff }  \alpha((p,\boldsymbol v),w)=\alpha((p',\boldsymbol v'),w')\in F',\]
as desired.
\end{proof}

\begin{theorem}\label{the:EntropyPDA}
Let $\mathcal{A}=(Q,\Sigma,\Gamma_1,\dots,\Gamma_k, \delta,q_0,F)$ be a deterministic $\varepsilon$-free $k$-stack PDA. Then 
\[\eta(L(\mathcal{A}))\leq \log_2|\Gamma_1|+\dots +\log_2|\Gamma_k|.\]
\end{theorem}
\begin{proof}
From Lemma \ref{lem:PDALambda} we know that every equivalence class in $\Lambda_n(\Top A)$ can be represented by a state of the form $(q,\boldsymbol{u})$, where $q\in Q$ and $\boldsymbol{u}\in \Gamma_1^{(n)}\times\dots\times\Gamma_k^{(n)}$. Therefore $\ind\Lambda_n(\Top A)\leq |Q|\cdot |\Gamma_1^{(n)}|\cdot\ldots\cdot |\Gamma_k^{(n)}|$.
Consequently,
\[\eta(L(\mathcal{A}))\leq\eta(\Top A)=\limsup_{n\to\infty}\frac{\log_2(\ind\Lambda_n(\Top A))}{n}\leq \log_2|\Gamma_1|+\dots +\log_2|\Gamma_k|.\]
This concludes the proof.
\end{proof}

We will see in Example \ref{exa:prodPali} that this upper bound can be reached using deterministic palindrome languages.
From this theorem the following conjecture from Schneider and Borchmann follows.
\begin{corollary}
Let $L$ be a language accepted by a deterministic $\varepsilon$-free $k$-counter automaton. Then 
\[\eta(L)=0.\]
\end{corollary}

Furthermore, we are now able to obtain a lower bound on the number of symbols in the stack alphabet.
\begin{corollary}\label{cor:lowerBoundStackSymbols}
Let $\mathcal{A}=(Q,\Sigma,\Gamma, \delta,q_0,F)$ be a deterministic $\varepsilon$-free $1$-stack PDA accepting the language $L$. Then $|\Gamma|\geq 2^{\eta(L)}$.
\end{corollary}

Now we discuss what happens if we drop some of the restrictions of determinisity and $\varepsilon$-freeness. In the following example we see that in most cases the entropy can no longer be bounded.
\begin{example}
Define the language
\[\T_\infty=\{a^{n_1}\texttt{\#}\dots \texttt{\#}a^{n_k}b^la^m\mid m,k,l,n_1,\dots,n_k\in\N,l\leq k, n_l=m\}.\]
To determine the entropy of $\T_\infty$ consider the set $\{a^{n_1}\texttt{\#}\dots \texttt{\#}a^{n_k}\mid k,n_1,\dots,n_k\leq n\}$. If two words $u$ and $u'$ from this set differ in the $i$\textsuperscript{th} block, i.e., $n_i\not=n'_i$. Then $b^ia^{n_i}$ is a word of length at most $2n$ that witnesses $(u,u')\notin\Theta_{2n}(\T_\infty)$. Therefore $\ind\Theta_{2n}(\T_\infty)\geq n^n$ and
\[\eta(\T_\infty) \geq \limsup_{n\to\infty}\frac{\log_2(n^n)}{2n}=\infty.\]
Clearly, there is a nondeterministic $\varepsilon$-free 1-stack PDA, a deterministic 1-stack PDA, and a nondeterministic $\varepsilon$-free 2-counter automaton accepting $\T_\infty$.
\end{example}

\begin{corollary}
	All deterministic $\varepsilon$-free context-free languages have finite entropy. But deterministic context-free languages can have infinity entropy.
\end{corollary}

The only thing left to study are 1-counter automata. The language $\T_\infty$ does not seem to be recognized by a 1-counter automaton. But the following modified version can be recognized by a nondeterministic one.
\[\B_k=\left\{ a_1^{n^1_1}\texttt{\#}\dots \texttt{\#}a_1^{n^1_{l_1}}\dots a_k^{n^k_1}\texttt{\#}\dots \texttt{\#}a_k^{n^k_{l_k}}ba_j^m\,\middle\vert\, \exists l.\ n^j_l=m\right\}\]

Consider the set
\[\left\{a_1^{n^1_1}\texttt{\#}\dots \texttt{\#}a_1^{n^1_{l_1}}\dots a_k^{n^k_1}\texttt{\#}\dots \texttt{\#}a_k^{n^k_{l_k}}\,\middle\vert\, n^i_1<n^i_2<\dots<n^i_{l_i}\leq n\text{ for all $1\leq i\leq k$}\right\}.\]
Any two words in this set can be separated by a witness of length at most $n$. Consequently, $\ind\Theta_n(\B_k)\geq 2^{nk}$ and
\[\eta(\B_k)\geq k.\]
Although the entropy of $\B_k$ is unfortunately not infinite we have at least found a family of languages with unbounded entropy. But beware we have increased the alphabet to obtain a higher entropy. The alphabet of $\B_k$ contains $k+2$ symbols. But if we encode $\B_k$ over a two element alphabet we obtain, by Corollary \ref{cor:encSameLength},
\[\eta(\enc(\B_k))\geq \frac{k}{\lceil\log_2(k+2)\rceil}.\]
This gives us a family with unbounded entropy over a fixed alphabet.

\begin{corollary}
Let $\Sigma$ be an at least two element alphabet.
The entropy of languages over $\Sigma$ accepted by nondeterministic 1-counter automata can not be uniformly bounded.
\end{corollary}

For deterministic counter automata that have at least two counters allowing $\varepsilon$-transitions makes the automata model Turing complete. We will show that if there is only one counter than $\varepsilon$-transitions might increase the computational expressibility but not the entropy.

For the computational expressibility consider the language
\begin{align*}
L=\{a^{n_1}i_1a^{m_1}\texttt{\#}\dots \texttt{\#}a^{n_k}i_ka^{m_k}\mid{}& k\geq1, n_1,\dots,n_k,m_1,\dots,m_k\in\N, 
i_1,\dots,i_k\in\{0,1\},\\
&\text{$i_l=1$ implies $n_l=m_l$ for all $l\leq k$}\}. %k\geq1
\end{align*}
There is a deterministic 1-counter automaton accepting $L$ but the $\varepsilon$-transitions seem to be absolutely necessary.

%To determine the entropy we classify what $\varepsilon$-transitions can be used for to replace them with some special (not $\varepsilon$-free) transitions. This will give us a slightly different $\Top A$ for which we can compute the entropy as before.
%Let us consider a maximal sequence of $\varepsilon$-transitions \[(q_1,n_1,w)\vdash(q_2,n_2,w)\vdash(q_3,n_3,w)\vdash \dots\]%, where we always assume that $n+k_1 +\dots+k_i>0$.
%Note that such a sequence is either finite or %periodic, i.e., $(q_1,n_1,w)=(q_k,n_1+m,w)$ for some $k,m$.
%where $n_1$ is minimal.
%There are four cases to consider:
To determine the entropy let us fix a deterministic 1-counter automaton $\mathcal{A}=(Q,\Sigma,\delta,q_0,F)$. And define $\Top A$ as 
$((Q\cup\{p_\infty\})\times\N^\infty,\Sigma,\alpha,(q_0,0),F\times\N^\infty)$,
where $\N^\infty=\N\cup\{\infty\}$, $p_\infty$ is a new state, and $\alpha$ is defined in the following way:
\[\alpha((q,c),w)=
\begin{cases}
(p,d)& \text{if $(q,c,w)\vdash^*(p,d,\varepsilon)$ and not $(p,d,\varepsilon)\vdash(p',d'
,\varepsilon)$}\\
%there is no transition starting with $(p,m,\varepsilon)$}\\
(p_\infty,\infty) &\text{otherwise}
\end{cases}
\]
%\[\alpha((q,n),w)$ is $(p,m)$ if $(q,n,w)\vdash^*(p,m,\varepsilon)$ and there is no transition starting with $(p,m,\varepsilon)$ or $(p_\infty,\infty)$ otherwise\]

Note that if $\mathcal A$ has no $\varepsilon$-transitions then $\Top A$ is the topological automation we constructed before.
To compute the entropy of $\Top A$ we first need to understand what transitions we added.
Fix $q\in Q$ and $a\in\Sigma$.
For every $c\in\N$ consider the unique maximal sequence 
\[(q,c,a)\vdash(q_1,c_1,w_1)\vdash(q_2,c_2,w_2)\vdash \dots\]
and define the function
%\[c\mapsto
%\begin{cases}
%(p_c,m_c)&\text{if the sequence is finite and $(p_c,m_c,\varepsilon)$ is its last element}\\
%(p_\infty,\infty)&\text{if the sequence is infinite}
%\end{cases}\]
$c\mapsto (p_c,m_c)$, where $\alpha((q,c),a)=(p_c,m_c)$.

Now there are two cases to consider. If there is a $C$ such that for all $i$ we have $c_i>0$. Then for all $c\geq C$ we have $m:=m_C-C=m_c-c$ and $p_c=p_C$ and $\alpha$ fulfils the following equation
\[\alpha((q,c),a)=
\begin{cases}
(p_c,m_c)&\text{if $c< C$}\\
(p_C,c+m)&\text{if $c\geq C$}
\end{cases}\]
Note that $m$ can be negative. %Choose $C$ minimal and define $C_{(q,a)}=\max\{C,m_c\mid c\leq C, m_c\neq\infty\}$. %Let $-M$ be the smallest $m_C-C$ occurring for any $(q,a)$.

In the other case we have for all $c$ that there is an $i$ with $c_i=0$. For larger $c$ the sequences become arbitrarily long and since there are only finitely many states there must be $j>i$ such that $q_i=q_j$, $w_i=w_j$, and $c_j<c_i$. Define $k=c_i-c_j$. Observe that the value of $\alpha((q,c),a)$ depends only on $c\bmod k$ and therefore \[\alpha((q,c),a)=(p_{c\bmod k},m_{c\bmod k}).\]
%Let $k$ be minimal and define $C_{(q,a)}=\max\{k,m_c\mid c < k, m_c\neq\infty\}$.

Let $\boldsymbol C$ be the maximum of all $C$, $\boldsymbol M\geq 1$ minimal such that no $m$ is smaller than $-\boldsymbol M$, and $K$ the set containing all $k$'s.  Now we can deduce the analogue of Lemma \ref{lem:PDALambda}.
\begin{lemma}\label{lem:OneCAEpsLambda}
If $c,c'>\boldsymbol M\cdot n\geq \boldsymbol C$, and $c\equiv c'\pmod k$ for all $k\in K$. Then $((q,c),(q,c'))\in\Lambda_n(\Top A)$ for all $q\in Q\cup\{p_\infty\}$.
\end{lemma}
\begin{proof}

We will proof the lemma  by induction on $n$.% that the positive witnesses of length at most $n$ are the same for both states of $\Top A$.

If $n=0$. Then the only witness to consider is $\varepsilon$ and the claim is trivial.

If $n>0$. Then consider a witness $aw$ of length at most $n$. %and let $\alpha((q,c),a)=(p,d)$ and $\alpha((q,c'),a)=(p',d')$. 
If $(q,a)$ gives a transition as in the first case, then let $C$ and $m$ be the parameters as above. Since $c,c'\geq \boldsymbol C\geq C$ we have %From the assumptions on $c$ and $c'$ we deduce that $p=p'$.
\begin{align*}
\alpha((q,c),a)=(p_C,c+m) && \text{and}&&\alpha((q,c'),a)=(p_C,c'+m).
\end{align*}
By assumption $m\geq -M$, hence $c+m,c'+m\geq \boldsymbol M\cdot(n-1)$. Also $c+m\equiv c'+m\pmod k$ for all $k\in K$. Therefore the claim holds by induction hypotheses.

If the transition is as in the second case with parameter $k$. Then, by assumption,
\[\alpha((q,c),a)=(p_{c\bmod k},m_{c\bmod k})=(p_{c'\bmod k},m_{c'\bmod k})=\alpha((q,c'),a)\]
and the claim holds trivially.
\end{proof}

Using this lemma we can deduce the desired theorem.
\begin{theorem}
Let $\mathcal A=(Q,\Sigma,\delta,q_0,F)$ be a deterministic 1-counter automaton with $\varepsilon$-transitions. Then
\[\eta(L(\mathcal{A}))=0.\]
\end{theorem}
\begin{proof}
Note that for every $c$ there is a $c'\leq \prod_{k\in  K}k$ such that $c\equiv c'\pmod k$ for all $k\in K$. Therefore Lemma \ref{lem:OneCAEpsLambda} implies that every class in $\Lambda_n(\Top A)$ contains a state where the counter value is at most $\max\{\boldsymbol M\cdot n,\prod_{k\in  K}k\}$. Hence for large $n$ 
\[\ind\Lambda_n(\Top A)\leq (|Q|+1)\cdot \boldsymbol M\cdot n\]
and $\eta(L(\mathcal{A}))=0$.
\end{proof}

In Table \ref{tab:upperBoundsUpdated} we summarize the results from this section.

\begin{table}[ht]
\centering
\begin{tabular}{l!{\vrule width 1.5pt}c|c!{\vrule width 1.5pt}c|c!{\vrule width 1.5pt}}
%\hline
%&$\varepsilon$-transitions, nondeterministic&$\varepsilon$-transitions, deterministic&$\varepsilon$-free, nondeterministic&$\varepsilon$-free, deterministic\\\hline
\multicolumn{1}{c!{\vrule width 1.5pt}}{}&\multicolumn{2}{c!{\vrule width 1.5pt}}{$\varepsilon$-transitions}&\multicolumn{2}{c!{\vrule width 1.5pt}}{$\varepsilon$-free}\\%\hline
\multicolumn{1}{c!{\vrule width 1.5pt}}{}&nondet.&deterministic&nondet.&deterministic\\\noalign{\hrule height 1.5pt}
$1$-counter automaton&$\infty\ (\uparrow\B_n)$&$0$&$\infty\ (\uparrow\B_n)$&$0$\\\hline
$1$-stack PDA&$\infty\ (\T_\infty)$&$\infty\ (\T_\infty)$&$\infty\  (\T_\infty)$&$\log_2|\Gamma|$\\\noalign{\hrule height 1.5pt}
$k$-counter automaton&$\underline{\infty}$&$\underline{\infty}$&$\infty\ (\T_\infty)$&$0$\\\hline
$k$-stack PDA&$\underline{\infty}$&$\underline{\infty}$&$\infty\ (\T_\infty)$&$\sum_{i=1}^k\log_2|\Gamma_i|$\\\noalign{\hrule height 1.5pt} %(\DetPali_{\Gamma_1\times\dots\times\Gamma_k})
\end{tabular}
\caption{Upper bounds for the entropy of languages accepted by certain kinds of automata. Here $k$ is at least 2. The $\underline{\infty}$ indicates that the model is Turing complete.}\label{tab:upperBoundsUpdated}
\end{table}

We have shown that most of the upper bounds presented in this table can be reached. But we suspect that this is not the case for all bounds.

\begin{conjecture}\label{con:oneCounterFiniteEntropy}
All languages accepted by 1-counter automata have finite entropy.
\end{conjecture}
This concludes our investigation of push-down automata.

\section{Entropy of Example Languages}\label{sec:Languages}

%\DanielSagt[inline]{Dies sollte ein eigenes Kapitel werden, oder?}

The purpose of this section is to give entropies of selected example languages. %Some of these are already known 
Some of these examples were already discussed in \cite{schneiderborchmann}, 
%Here we will present the results of examples from \cite{schneiderborchmann} and calculate the entropy of some languages. The 
the results from Examples \ref{exa:Dyck2}, \ref{exa:DetPali}, and \ref{exa:prime} are new and have to our knowledge not been discussed before.

\begin{example}
Let $L$ be a regular language. Then $\Theta(L)$ is finite. Hence
\[\eta(L)=\limsup_{n\to\infty}\frac{\log_2\ind\Theta_n(L)}{n}\leq \limsup_{n\to\infty}\frac{\log_2\ind\Theta(L)}{n}=0.\]
All regular languages have zero entropy, which goes with the intuition that regular languages are simple.
\end{example}

%Conversely, in \cite{schneiderborchmann} we have already seen that the standard example of a context-sensitive language $\{a^nb^nc^n\mid n\in\N\}$ also has zero entropy. Hence not every language with zero entropy is also regular.

%Before we come to the next example, note that for any language $L$, $n\in\N$, and any word $w$ the equivalence class %\FlorianSagt{ist das wirklich eine kongruenzklasse? Ist $\Theta_n$ eine kongruenzrelation?}
%of $\Theta_n(L)$ generated by $w$ is characterized by the set $U_w=\{v\in\Sigma^{(n)}\mid wv\in L\}$ consisting of all positive witnesses of length at most $n$. Since for all $w,w'\in\Sigma^*$
%\begin{align*}
%[w]=[w']&\iff (wv\in L\iff w'v\in L)\text{ for all $v\in\Sigma^{(n)}$}\\
%&\iff (v\in U_w\iff v\in U_{w'})\text{ for all $v\in\Sigma^{(n)}$}\\
%&\iff U_w=U_{w'}.
%\end{align*}

%$\{w\mid |w|_a=|w|_b\}$
%Now we come to languages which have an entropy greater than zero.
\begin{example}\label{exa:Dyck2}
An example considered in \cite{schneiderborchmann} is the \define{Dyck language with $k$ sorts of parenthesis}, which consists of all balanced strings over $\{(_1,)_1,\dots,(_k,)_k\}$.
More generally, let $\Gamma$ be an alphabet and $\overline{\Gamma}=\{\overline{a}\mid a\in\Gamma\}$. Then $\overline{\phantom{a}}\colon\Gamma\to\overline{\Gamma}$ is a bijection. %, we denote the inverse of $\overline{\phantom{a}}$ again by $\overline{\phantom{a}}$. 
Now the \define{Dyck language over $\Gamma$}, denoted by $\Dyck_\Gamma$, is the set of all words $u$ such that successively replacing $a\overline{a}$ in $u$ by $\varepsilon$ results in $\varepsilon$.

\begin{lemma}\label{lem:Dyck}
For all alphabets $\Gamma$ we have $\eta(\Dyck_\Gamma)=\log_2|\Gamma|$.
\end{lemma}
\begin{proof}
%$\Sigma=\Gamma\cup\overline{\Gamma}$, $|\Gamma|=k$
%Note that $\Dyck_k=\{w\in\Sigma^*\mid \red(w)=\varepsilon\}$. Two words $u$ and $v$ are in $\Theta_n$ if and only if $\red(u)$ and $\red(v)$ are equal and in $\Gamma^{(n)}$ or if both $\red(u)$ and $\red(v)$ are not in $\Gamma^{(n)}$. In the first case there are $w\in\Sigma^{(n)}$ such that $vw,uw\in\Dyck_k$, in the latter case $vw,uw\notin\Dyck_k$ for all $w\in\Sigma^{(n)}$. 
Let $a\in\Gamma$. Observe that the set $\Gamma^{(n)}\cup\{\overline{a}\}$ contains exactly one representative of each class of $\Theta_n(\Dyck_\Gamma)$. Therefore $\ind\Theta_n(\Dyck_\Gamma)=|\Gamma^{(n)}|+1$ and we can compute
\begin{align*}
\eta(\Dyck_\Gamma)=\limsup_{n\to\infty}\frac{\log_2 (|\Gamma^{(n)}|+1)}{n}=\log_2|\Gamma|.
\end{align*}
This finishes the proof.
\end{proof}
\end{example}

Another example discussed in \cite{schneiderborchmann} is the palindrome language. For an alphabet $\Sigma$ define the \define{palindrome language over $\Sigma$} to be
\[\Pali_\Sigma=\{uu^R \mid u\in\Sigma^*\}.\]
Schneider and Borchmann showed that $\log_2|\Sigma|\leq \eta(\Pali_\Sigma)\leq \log_2|\Sigma|+1$.

\begin{example}\label{exa:DetPali}
We will consider here the \define{deterministic palindrome language}. Assume that $\texttt{\#}$ is not in $\Sigma$ and define
\[\DetPali_\Sigma=\{u\texttt{\#}u^R\mid u\in\Sigma^*\}.\]
The entropy can be computed similar to that of the Dyck languages. Consider a class in $\Theta_n(\DetPali_\Sigma)$ with at least one positive witness. Let $w$ be the positive witness of minimal length. Then either $w=\texttt{\#}w'$ or $w\in\Sigma^*$. In the first case the class is represented by $w'^R$ and in the second case it is represented by $w^R\texttt{\#}$. Therefore
%We will show that as in the nondeterministic case, $\log_2|\Sigma|$ is a lower bound for the entropy of $\DetPali_\Sigma$. Consider the set $\Sigma^{n}$ of words of length $n$. For two words $u,v\in\Sigma^n$ we have that $u\texttt{\#}v^R\in\DetPali_\Sigma$ iff $u=v$. Since $v\texttt{\#}v^R\in\DetPali_\Sigma$, every word in $\Sigma^n$ is in a different class of $\Theta_{n+1}(\DetPali_\Sigma)$. Hence $|\Sigma|^n$ is a lower bound for $\ind\Theta_{n+1}(\DetPali_\Sigma)$ and 
\[\eta(\DetPali_\Sigma)=\limsup_{n\to\infty}\frac{\log_22\cdot|\Sigma^{(n)}|}{n}=\log_2|\Sigma|.\]
\end{example}

Both $\Dyck_\Gamma$ or $\DetPali_\Gamma$ are typical examples for languages recognized by PDAs. We can apply Corollary \ref{cor:lowerBoundStackSymbols} to obtain the following result.

\begin{corollary}
Every deterministic $\varepsilon$-free $1$-stack PDA recognizing $\Dyck_\Gamma$ or $\DetPali_\Gamma$ has at least $|\Gamma|$ many stack symbols.
\end{corollary}

\begin{example}\label{exa:prodPali} We construct a product palindrome language.
	Let $\Gamma_1,\dots,\Gamma_k$ be alphabets %$L=\{w\in(\bigcup\Gamma_i)^*\mid \pi_{\Gamma_i}(w)\in\mathbf{Pali}_{\Gamma_i}\}$.
	%\[L=\{w\in((\Gamma_1\cup\{\sharp, \textvisiblespace\})\times\dots\times(\Gamma_k\cup\{\sharp, \textvisiblespace\}))^*\mid \pi_{\Gamma_i\cup\{\sharp\}}(P_{i}(w))\in\mathbf{DPali}_{\Gamma_i}\}\]
	%where $P_i(w)$ is the $i$\textsuperscript{th} row of $w$ and $\pi_{\Gamma_i\cup\{\sharp\}}(u)$ just removes all $\textvisiblespace$ from $u$. So for example the word
	%{\setlength\tabcolsep{0pt}
	%\begin{tabular}{ccccc}
	%$a$&$b$&$\sharp$&$b$&$a$\\
	%$a$&$\sharp$&$a$&\textvisiblespace&\textvisiblespace
	%\end{tabular}} could be in $L$.
	and consider the language $\DetPali_{\Gamma_1\times\dots\times\Gamma_k}$. As we have just seen the entropy of this language is
	\begin{align*}
	\eta(\DetPali_{\Gamma_1\times\dots\times\Gamma_k})=\log_2|\Gamma_1|+\dots+\log_2|\Gamma_k|.
	\end{align*}
\end{example}

Obviously, there is a deterministic $\varepsilon$-free $k$-stack PDA with stack alphabets $\Gamma_1,\dots,\Gamma_k$ accepting $\DetPali_{\Gamma_1\times\dots\times\Gamma_k}$. Therefore the upper bound given in Theorem \ref{the:EntropyPDA} can be reached.

%We have shown that there are languages of arbitrarily large entropy, so the next question to ask is: are there languages with infinite entropy?

Finally, we will discuss the mathematically very interesting language of all prime numbers.
\begin{example}\label{exa:prime}
Let $\Prime=\{p\mid \text{$p$ is prime}\}$ be the unary encoding of all prime numbers.
For $m>2$ there cannot be two consecutive numbers in $\Pos_m$ since the only $k$ for which $k$ and $k+1$ is prime is $2$. Because of this, $\Pos_m$ contains only even numbers or only odd numbers if $m$ is odd or even, respectively. Consequently, $\ind\Theta_n(\Prime)\leq 2\cdot 2^{\lceil n/2\rceil}+3$ and $\eta(\Prime)\leq \frac{1}{2}$. We can make a similar argument for the prime 3.  For $m>3$ there cannot be a $k$ such that $k,k+2,k+4\in \Pos_m$ because one of these numbers is divisible by 3. 
Thus $\ind\Theta_n(\Prime)\leq 2\cdot3\cdot 2^{\lceil \frac{n}{2}\cdot\frac{2}{3}\rceil}+4$ and $\eta(\Prime)\leq\frac{1}{3}$. These observations lead to the following definition:
\begin{definition}
A set $A\subseteq\{0,\dots,n-1\}$ \define{represents a plausible sequence of primes of length $n$} if for all primes $p\leq n$
\[A\bmod p\neq \{0,\dots,p-1\}.\]
The set $A$ \define{represents an occurring sequence of primes of length $n$} if there is a $k\in\N$ such that for all $l\in\{0,\dots,n-1\}$
\[k+l\text{ is prime}\iff l\in A.\]
We denote the number of plausible sequences of length $n$ by $s_{n}$.
\end{definition}
Note that for example the set $\{0,1\}$ occurs for $k=2$, but is not plausible since $\{0,1\}\bmod2=\{0,1\}$. This can happen because the idea behind the plausible sequences is that every number in the sequence dividable by $p$ is not a prime, which holds for all numbers except for $p$ itself. But if we say that the starting point $k$ of the sequence is greater than the length $n$ of the sequence, than $p$ cannot occur in $k+\{0,\dots,n-1\}$. Hence if $A$ represents an occurring sequence of length $n$, which occurs for some $k>n$, then $A$ also represents a plausible sequence. %Whether the converse holds is unknown to the author.

With this observation we can bound $\ind\Theta_n(\Prime)$ by 
\[|\{[k]\mid k\leq n+1\}|+|\{[k]\mid k> n+1\}|\leq n+2+s_{n+1}.\]

\begin{lemma}\label{lem:primeBound}
Let $p_1,\dots,p_k$ be the first $k$ primes and $n=n'\cdot\prod_{i\leq k}p_i$. We have that 
\[s_n\leq\left(\prod_{i\leq k}p_i\right)\cdot2^{n\cdot\prod_{i\leq k}\frac{p_i-1}{p_i}}.\]
\end{lemma}
\begin{proof}
For every $i\leq k$ fix an $\ell_i\in\{0,\dots,p_i-1\}$. Now we bound the number of plausible sequences $A$ of length $n$ with $\ell_i\notin (A\bmod p_i)$ for all $i$.
Clearly,
\[A\subseteq \bigcap_{i\leq k}\left(\{0,\dots,n-1\}\setminus\{m\cdot p_i + \ell_i\mid m<\frac{n}{p_i}\}\right).\]
Furthermore,
\begin{align*}
\left|\{0,\dots,n-1\}\setminus\{m\cdot p_i + \ell_i\mid m<\frac{n}{p_i}\}\right|=n-\frac{n}{p_i}
\end{align*}
and
\begin{align*}
\left|\bigcap_{i\leq k}\left(\{0,\dots,n-1\}\setminus\{m\cdot p_i + \ell_i\mid m<\frac{n}{p_i}\}\right)\right|&=\sum_{m, j_1,\dots,j_m\leq k}(-1)^m\frac{n}{p_{j_1}\cdot\ldots\cdot  p_{j_m}}\\
&={n\cdot\prod_{i\leq k}\frac{p_i-1}{p_i}}.
\end{align*}
Together with the fact that there are $\prod_{i\leq k}p_i$ possible choices for $\ell_1,\dots,\ell_k$ this proves the claim.
%$\Pos_u^{(n)}$ $\Pos_u^{n}$ $\Pos(u,n)$ $\Pos^=(u,n)$
%First consider the case where $n'=1$. Induction on $k$. Let $\{n_1,\dots, n_l\}$ be the plausible positions in $\prod_{p\leq k}p$ and $q$ the next prime. Then in $q\cdot\prod_{p\leq k}p$ the $k$-plausible positions are $\{n_1,\dots, n_l,\dots,n_1\cdot q,\dots,n_l\cdot q\}$. As $q$ is coprime to $\prod_{p\leq k}p$ we have that there are $(q-1)\cdot l$ $q$-plausible position.
%w.l.o.g. we want to find an $A\subseteq \{0,\dots,n-1\}$ such that $A \bmod p\subseteq \{1,\dots,p-1\}$.
%\[\bigcup A\cap\{0,\dots,p-1\}+k\cdot p\]
%Let us look at the set $A$ of available positions up to $k$, i.e., $A\subseteq\{0,\dots,k-1\}$ maximal such that $A_k\bmod p\subset\{1,\dots,p-1\}$. We show by induction on $n$ that $|A_{p_1\cdot\ldots\cdot p_n}|=\frac{p_1-1}{p_1}\cdot\ldots\cdot\frac{p_n-1}{p_n}$.
%If $n=0$, then the claim is trivially true.
%If $n>0$, then we can apply the induction hypothesis to $A_{p_1\cdot\ldots\cdot p_{n-1}}$ now the claim follows from the fact that $p_1\cdot\ldots\cdot p_{n-1}$ and $p_n$ are coprime.
\end{proof}

Using this lemma we can compute the entropy.

\begin{theorem}
The entropy $\eta(\Prime)$ is zero.
\end{theorem}
\begin{proof}
From the previous observation and Lemma \ref{lem:primeBound} we can deduce that 
\begin{align*}
\eta(\Prime)\leq \prod_{i\leq k}\frac{p_i-1}{p_i}
\end{align*}
for all $k$.
An old result from Mertens \cite{mertens} states \[\prod_{i\leq k}\frac{p_i-1}{p_i}\in \mathcal O \left(\frac{1}{\log_2(k)}\right).\text{ Therefore }\lim_{k\to\infty} \prod_{i\leq k}\frac{p_i-1}{p_i}=0\] and we conclude that $\eta(\Prime)=0$. 
\end{proof}

Even though not relevant for the entropy of $\Prime$, the following observations are simply too beautiful not to be mentioned.
We conjecture the following.
\begin{conjecture}\label{con:primeSequences}
Every plausible sequence of primes occurs at least once.
\end{conjecture}
We have computationally verified this conjecture for sequences of length up to 29.
What makes this conjecture interesting is the following lemma.
\begin{lemma}
The following statements are equivalent:
\begin{enumerate}
\item Every plausible sequence of primes occurs at least once.
\item Every plausible sequence of primes occurs infinitely often.
\end{enumerate}
\end{lemma}
\begin{proof}
Assume the sequence represented by $A\subseteq\{0,\dots,n\}$ occurs only $k$ times. Then let $N\in\N$ such that there are more than $k$ plausible sequences represented by $A_1,\dots,A_l$ of length $N$ with the initial sequence as prefix, i.e., $A_i\cap \{0,\dots,n\}=A$. Hence at least one of these $A_i$ does not represent an occurring sequence, a contradiction.
\end{proof}
Note that the twin prime conjecture can be formulated as: the sequence $\{0,2\}$ occurs infinitely often. %...\FlorianSagt{some more text}
As a consequence, Conjecture \ref{con:primeSequences} is a generalization of the twin prime conjecture and the Green-Tao Theorem.

\end{example}

\section{Conclusion}

In this article we introduced the notion of a topological automaton from Steinberg~\cite{steinberg}. We defined the topological entropy of a topological automata and of a formal language.

We further investigated the notion of topological entropy of formal languages and its suitability as a measure of the complexity of formal languages. 
We were able to calculate the entropy of the Dyck languages, an previously open problem, and provided many other new examples.

We modified an example from \cite{schneiderborchmann} to show that the entropy function
is surjective for every $\Sigma$ with $|\Sigma|\geq 2$. Whether this is also the case for unary alphabets remains an open problem (Conjecture \ref{con:surjUnary}).%We showed that 1 is an upper bound for the entropy of unary languages. Naturally, the question arises whether there the entropy of a unary language can have every value in $[0,1]$. We believe this to be the case but were not able to prove it. %je nachdem nur für rationale zahlen
% We were able to solve the previously open problem of the entropy of the Dyck languages. 

Our second main result concerns a conjecture from Schneider and Borchmann \cite{schneiderborchmann}.
They suspected that all languages accepted by a one-way finite automaton equipped with a fixed number of counters and an acceptance condition that does only require to check local conditions have zero entropy. We showed that this conjecture holds if we assume the automaton to be deterministic and $\varepsilon$-free and we were even able to generalize this result to deterministic $\varepsilon$-free push-down automata. We showed that the entropy of a language accepted by such an automaton is bounded in terms of the sizes of the stack alphabets of the automaton. This result proves that all deterministic $\varepsilon$-free context-free languages have finite entropy. An open problem from this section is whether all languages accepted by nondeterministic 1-counter automata with $\varepsilon$-transitions have finite entropy (Conjecture \ref{con:oneCounterFiniteEntropy}).

On the other hand, we also saw that the definition of entropy is not very robust, since we can use padding to decrease the entropy of any language to zero. Consequently, there are also undecidable languages with zero entropy. It is also counterintuitive that the entropy of a language is not the same as the entropy of the reversed language. %\FlorianSagt{TODO show this somewhere, maybe in state of the art?}
Hence we suggest to define something like the entropy of the \define{core} of a language with the following properties:
\begin{itemize}
\item the entropy of the core of a language is at least as large as the entropy of the language,
\item padding a language does not influence the entropy of the core of this language,
\item reversing a language does not influence the entropy of the core of the language,
\item the entropy of the core of the languages we used to show surjectivity should be infinite, and
\item encoding the language should not change the entropy of the core of the language.
\end{itemize}

We propose to define the entropy of the core of a language $L$ in the following way:
\[\eta_{\core}(L)=\sup \{\eta(L')\mid L'\in\core(L)\},\]
where $\core(L)$ should contain at least $L$, $L^R=\{w^R\mid w\in L\}$, and every $L'$ such that there is an encoding $\enc$ with $\enc(L')=L$. But a suitable definition of $\core(L)$ remains to be found.

\end{document}